\documentclass[12pt]{article}

\usepackage{times}
\usepackage{fullpage}
\usepackage{amsfonts,amssymb,amsmath,amsthm}
\usepackage{latexsym} 
\usepackage{gauss}
\usepackage{tikz}
\usepackage{url}
\usepackage{thmtools}
\usepackage{thm-restate}
\usepackage{hyperref}
\usepackage[capitalise, nameinlink]{cleveref}

\newcommand{\R}{\mathbb{R}}
\newcommand{\T}{\mathrm{T}}

\newcommand{\ceil}[1]{\left\lceil #1 \right\rceil}

\def\01{\{0,1\}}
\newcommand{\q}{\tilde{q}}
\newcommand{\er}{\tilde{r}}
\newcommand{\z}{\tilde{z}}
\newcommand{\Z}{\tilde{Z}}

\newcommand{\eT}{\tilde{T}}
\newcommand{\eR}{\tilde{R}}

\newcommand{\br}{\mathrm{br}}
\newcommand{\diag}{\mathrm{diag}}
\newcommand{\cp}{\mathrm{cp}}
\newcommand{\e}{\mathrm{e}}

\newtheorem{theorem}{Theorem}
\newtheorem{question}[theorem]{Question}
\newtheorem{lemma}[theorem]{Lemma}
\newtheorem{corollary}[theorem]{Corollary}

\newtheorem{claim}[theorem]{Claim}

\theoremstyle{definition}

\newtheorem{definition}[theorem]{Definition}

\begin{document}
\title{Strategies for quantum races}
\author{Troy Lee \thanks{Centre for Quantum Software and Information, School of Software, Faculty of Engineering and Information Technology, University of Technology Sydney, Australia (\tt{troyjlee@gmail.com})} \and Maharshi Ray \thanks{Centre for Quantum Technologies, National University of Singapore, Singapore (\tt{maharshi91@gmail.com})} \and Miklos Santha 
\thanks{IRIF, Univ.\  Paris Diderot, CNRS, 75205 Paris, France;  and
Centre for Quantum Technologies and MajuLab, National University of Singapore, Singapore ({\tt miklos.santha@gmail.com})}}
\maketitle

\begin{abstract}
We initiate the study of \emph{quantum races}, games where two or more quantum computers compete to solve a 
computational problem.  While the problem of dueling algorithms has been studied for classical deterministic algorithms 
\cite{IKLMPT11}, the quantum case presents additional sources of uncertainty for the players.  The foremost among these is that players do 
not know if they have solved the problem until they measure their quantum state.  This question of ``when to measure?'' presents a very 
interesting strategic problem.  We develop a game-theoretic model of a multiplayer quantum race, and find an approximate 
Nash equilibrium where all players play the same strategy.  In the two-party case,
we further show that this strategy is nearly optimal in terms of payoff among all symmetric Nash equilibria.  A key role in our 
analysis of quantum races is played by a more tractable version of the game where there is no payout on a tie; 
for such races we completely characterize the Nash equilibria in the two-party case.

One application of our results is to the stability of the Bitcoin protocol when mining is done by quantum computers.  Bitcoin mining is a 
race to solve a computational search problem, with the winner gaining the right to create a new block.  Our results inform the strategies 
that eventual quantum miners should use, and also indicate that the collision probability---the probability that two miners find a new block 
at the same time---would not be too high in the case of quantum miners.  Such collisions are undesirable as they lead to forking of the Bitcoin blockchain.
\end{abstract}

\section{Introduction}
We study the scenario of two or more quantum computers competing to solve a computational task, which we call 
a \emph{quantum race}.  This setting presents a different problem to finding the fastest algorithm for a task, as the 
only goal is to solve the task before the competition.  
For example, imagine a search race where Alice and Bob, each armed with identical quantum computers, compete to find a marked item in a 
database.  The first person to find the marked item wins \$1, with the payout being split in the case of a tie.  The first natural idea is for 
Alice to run Grover's algorithm \cite{Grover:1996}, which can find a marked item in a database of size $N$ with high probability in time $O(\sqrt{N})$. However, if Alice's 
strategy is to run Grover's algorithm and measure after the specified number of steps to maximize her success probability, Bob 
will have an advantage by measuring after running Grover's algorithm for a few less steps.  Although this way Bob
has a slightly lower success probability, he gains a huge advantage in always answering first.  This simple example 
shows that the optimal algorithm to solve a problem can be different from the optimal strategy to employ when the goal is to 
solve the problem before an opponent. 

The scenario of competing algorithms has been studied before in the classical deterministic setting \cite{IKLMPT11}.  In a classical game,
the uncertainty is provided by an unknown probability distribution over the inputs: depending on what the input is, one 
algorithm may perform better than another.  The quantum setting inherently has additional sources of uncertainty, most interestingly that 
players do not know if they have solved the problem until they \emph{measure} their quantum state.  Going back 
to the search game, in the classical version the players know at every instant if they have found the marked item or not.
This is not the case in the quantum setting, where a player can only tell if she has found the desired item by 
measuring her quantum state. Furthermore, if she measures her state and does not find the marked item, then she must begin the 
search again from scratch.  In the quantum case there is a natural tension between waiting to measure, and thereby building up the probability of success 
upon measuring, and measuring sooner, to answer before one's competitors.  We study this game theoretic problem to develop 
strategies for players to use in quantum races.

One of our main motivations for studying quantum races is to model quantum computers mining the decentralized currency Bitcoin \cite{Nak09}.  
Mining is the process by which new blocks of transactions are added to the history of Bitcoin transactions, called the blockchain.  
The winner of a race to solve a computational search problem gains the right to add a new block of transactions to the blockchain, and participants in 
this race are called miners.  Quantum miners could use Grover's algorithm 
to solve the search problem with quadratically fewer search queries than needed classically.  But what should the strategy of quantum miners be when competing against each other?  
\begin{question}
\label{question1}
What is the optimal strategy for quantum miners?
\end{question}
Figuring out the optimal strategy for quantum miners is important to analyze the impact of quantum mining on the stability of the Bitcoin protocol.  
When two miners solve the computational search problem at (nearly) the same time, the blockchain can fork as it is unclear which new block is the 
``correct'' history of Bitcoin transactions.  Forking is bad for the security of Bitcoin as it can decrease the cost of attacks \cite{GKW+16}, increase the gain from deviating from the 
intended mining protocol \cite{ES16}, and generally decreases chain growth and wastes resources.   In the classical case, each search query has the same probability of 
success.  In the quantum case, however, because of Grover's algorithm the success probability grows roughly quadratically with the number of search queries.  
Does this lead to many quantum miners finding blocks at the same time?
\begin{question}
\label{question2}
What is the probability that two or more quantum miners playing the optimal strategy find a block at the same time?
\end{question}
In the next subsections, we describe our model and results in more detail and the impact it has on these questions.

\subsection{The model and results}
In a \emph{symmetric} game all players have the same payoff function.  In his original paper defining a Nash equilibrium, Nash 
showed that every symmetric multiparty finite game has a symmetric equilibrium, i.e.\  one where all players play the same strategy \cite{Nash1951}.  
When all players have identical quantum computers, a quantum race is naturally a symmetric race, and we describe this scenario first.

We model a symmetric multiplayer quantum race in the following way.  The pure strategies available to a player are the possible times 
at which she can measure $1, 2, 3, \ldots, K$.  For each time $t$, a player has an algorithm that she can run for $t$ 
steps and for which the success probability is $p_t$.  Without loss of generality, we assume that these probabilities 
form an increasing sequence $0<p_1 < p_2 < \cdots < p_K \le 1$.  A general strategy is a probability distribution 
over the possible times to measure.  The player who succeeds first receives a payoff of 1.  In the case of a tie, the payoff is 
split amongst all players who succeed first at the same time.  Our model can be thought of as a ``one-shot'' race, as if a 
player measures and does not succeed, she does not get a chance to restart and try again.  While a race where players are allowed 
to repeatedly restart until someone wins would be more realistic, it becomes much more difficult to analyze due to the proliferation of possible stategies, and 
we leave this for future work.  

\paragraph{Two-player case}
We begin explaining our model and results in more detail in the two-player case.  In this case, a game defined by the 
probabilities $p_1, \ldots, p_K$ can be represented by the payoff matrix for Alice, given by the $K$-by-$K$ matrix $A$, and the payoff matrix for Bob $B$.  
The $(s,t)$ entry of $A$ gives Alice's payoff when she runs an algorithm for time $s$ and Bob plays time $t$.  In the case of a quantum race, this is defined 
as
\begin{equation}
\label{eq:intro_payoff}
A(s,t) = 
\begin{cases}
p_s & \mbox{ if } s < t \\
p_s (1-p_s) + \frac{1}{2} p_s^2 & \mbox{ if } s = t \\
p_s (1-p_t) & \mbox{ if } s > t \enspace .
\end{cases}
\end{equation}
As the game is symmetric, Bob's payoff matrix is $B=A^\T$.  

Our analysis of quantum races begins in \cref{sec:too_stingy} by analyzing a more tractable variant of the game we call a \emph{stingy} quantum race.  
In a stingy quantum race, there is no payout in the case of a tie (the game organizer is stingy).  A Nash equilibrium of a 
two-party stingy quantum race has very strong constraints on its support structure (see \cref{cor:support_structure}).  In particular, if $(x,y)$ is a 
Nash equilibrium in a two-party stingy quantum race, then the union of the supports of $x$ and $y$ must be an interval $\{T, T+1, \ldots, K\}$ that 
contains the maximum running time $K$.  There are 3 possible types of Nash equilibria in a two-party stingy quantum race, and we characterize all of them
(see \cref{thm:symmetric}, \cref{thm:alternating}, and \cref{thm:alt_coinc}).  

One particularly nice type of equilibrium is what we call a \emph{coinciding} equilibrium.  In a coinciding equilibrium, the support of all player's strategies is the same, 
but the strategies do not have to be identical.  This is a more general notion than a symmetric equilibrium where all strategies are the same.  In a coinciding equilibrium 
for a stingy quantum race, the support of each player's strategy is an interval $\{T, T+1, \ldots, K\}$.
This leaves the problem of determining the starting point $T$ of this interval in a Nash equilibrium.  We are able to show that there is always exactly 
one $T$ such that there is a Nash equilibrium with support $\{T, T+1, \ldots, K\}$.

\begin{theorem}[Informal, see \cref{thm:symmetric}]
In a two-party stingy quantum race defined by probabilities $p_1, \ldots, p_K$, there is a unique coinciding Nash equilibrium.  In this 
equilibrium all players play the same strategy, and the support of the strategies is an interval $\{T^*, T^*+1, \ldots, K\}$.
\end{theorem}
We also explicitly find this Nash equilibrium.

This result begs the question: what is this value of $T^*$?  At what success probability does it become worthwhile to start measuring?  By putting an 
additional restriction on the probabilities $p_1, \ldots, p_K$, we can give quite a precise answer to this question.  
We say $p_1, \ldots, p_K$ is an $\ell$-\emph{dense} sequence (see \cref{def:dense}) if $p_1 \le \frac{\ell}{K}, p_K \ge 1-\frac{\ell}{K}$, and 
$p_{i+1} - p_i \le \frac{\ell}{K}$ for $i=2, \ldots, K-1$.  This is quite a natural restriction that is satisfied 
for many races.  In the quantum search race,
where the $p_i$ are the Grover success probabilities, and therefore also in the application to Bitcoin, the $\ell$-density condition
is satisfied with $\ell = \pi/2$.  In the $\ell$-dense case, we can give the following bound on $T^*$. 
\begin{theorem}[Informal, see \cref{cor:ptstar} and \cref{cor:payoff}]
\label{thm:intro1}
Let $p_1, \ldots, p_K$ be an $\ell$-dense sequence with $K \ge 6\ell$.  Then the starting point $T^*$ of the unique coinciding Nash equilibrium in the 
stingy quantum race defined by these probabilities is such that $p_{T^*}= \sqrt{2}-1 + \Theta\left(\frac{\ell}{K} \right)$.
\end{theorem}
Thus it is worthwhile to start measuring once the success probability becomes around $\sqrt{2}-1$, and this is largely independent of the actual values
of $p_1, \ldots, p_K$.  

In \cref{sec:two_races}, we apply our analysis of two-player stingy quantum races to the case of general quantum races.   
As the only difference between a stingy quantum race and a quantum race is the payout on ties, intuitively strategies in these two kinds of 
races should have similar payoffs when the probability of ties is small.  We follow this intuition and show that when 
$p_1, \ldots, p_K$ form an $\ell$-dense sequence the probability of a tie is $O(\frac{\ell}{K})$ (see \cref{thm:col}) when players use the unique 
coinciding equilibrium from the stingy race, and this strategy is an $O(\frac{\ell}{K})$-approximate Nash equilibrium of the corresponding quantum race.

Approximate Nash equilibria are naturally an imperfect lens into true Nash equilibria.  The approximate Nash equilibrium we give would not 
be a reasonable suggestion for the actual strategies of quantum players if there were other equilibria with much higher payoff, for example.  We 
show that this is not the case, and the approximate Nash equilibrium we give is nearly optimal in terms of payoff amongst all 
symmetric equilibria. 
\begin{theorem}[Informal, see \cref{thm:two_approx} and \cref{thm:payoff_upper}]
\label{thm:intro2}
Let  $p_1, \ldots, p_k$ be an $\ell$-dense sequence with $K \ge 6\ell$.  Then the unique coinciding equilibrium of the 
two-player stingy quantum race defined by these probabilities is an $O(\frac{\ell}{K})$-approximate Nash equilibrium in the corresponding quantum race.  
Moreover, the payoff achieved by this strategy is within $O(\sqrt{\frac{\ell}{K}})$ of the largest payoff achievable by any symmetric Nash equilibrium.  
\end{theorem}

To show that the approximate Nash equilibrium we give is nearly optimal in terms of payoff (\cref{thm:payoff_upper}), we use the bilinear programming formulation 
of Nash equilbria due to Mangasarian and Stone \cite{MS1964}.  Luckily, in our case the sum of Alice's and Bob's payoff matrices $A + A^\T$ (from \cref{eq:intro_payoff}) turns out to 
be negative semidefinite.  When optimizing over symmetric strategies this makes the Mangasarian and Stone bilinear program (which is a maximization problem) into a convex quadratic program.  
We then use Dorn's \cite{Dorn1960} equivalent dual formulation of a convex quadratic program (see \cref{eq:dual}), which is a minimization problem.  
We explicitly construct a feasible solution to this dual minimization problem to upper 
bound the payoff of any symmetric Nash equilibrium.  Our construction of this dual solution again makes use of our analysis of stingy quantum races.

\paragraph{Multiplayer case}
The case of many players is what we are interested in for the application to Bitcoin.  Luckily, we are able to recover analogs of many of the results from the two-player case in the 
multiplayer case as well.  We start in \cref{sec:multi} by analyzing $n$-player stingy quantum races, and show the following.
\begin{theorem}[Informal, see \cref{thm:msymmetric}]
\label{thm:m_intro}
Let $p_1, \ldots, p_K$ define an $n$-player stingy quantum race.  This race has a unique coinciding Nash equilibrium, and in this
equilibrium all players play the same strategy.  The support of each strategy is an interval $\{T^*, T^*+1, \ldots, K\}$.   
\end{theorem}
To show that an $n$-player stingy quantum race has a unique coinciding Nash equilibrium, our proof proceeds through 
a 2-player \emph{asymmetric} stingy quantum race.  In a 2-player 
asymmetric race, Alice has probabilities $p_1, \ldots, p_K$ of succeeding after $t$ steps and Bob has a (potentially different) 
sequence of probabilities $P_1, \ldots, P_K$.  An asymmetric race models the case where Alice and Bob have quantum 
computers of potentially different speeds.  We relate the payoff for Alice in a $n$-player stingy quantum race
to the payoff for Alice in a 2-player quantum race against a more powerful opponent (see \cref{lem:reduction}).  We can then 
refer to \cref{thm:unique1-5} in \cref{sec:too_stingy} which completely characterizes
coinciding equilibria in asymmetric 2-player stingy quantum races.  This gives \cref{thm:m_intro}.   

When the sequence $p_1, \ldots, p_K$ is $\ell$-dense, we can also say something about the starting point $T^*$ of the $n$-player coinciding 
equilibrium, though not as precisely as in the two-party case.
\begin{theorem}[Informal, see \cref{lem:bound_PTstar}]
\label{lem:bound_PTstar}
Let $p_1, \ldots, p_K$ be an $\ell$-dense sequence defining a stingy $n$-player quantum race with $n \ge 2$.  If $K \ge 4 \e \ell n$ then the starting point $T^*$ of 
the unique coinciding equilibrium is such that $p_{T^*} = \Theta(\frac{1}{n})$.
\end{theorem}
This means that the more players there are in a game, the earlier one starts to measure in the unique coinciding equilibrium.  

In light of \cref{question2} we also want to see what the probability of more than one player succeeding at the same time in 
this unique coinciding equilibrium.  We show the following.
\begin{restatable}{theorem}{colprobm}
\label{colprobm}
Let $P_1,....,P_K$ define an $\ell$-dense stingy $n$-player quantum race such that $4\e n \ell \le K$. 
When the players play the coinciding equilibrium of the stingy race, the probability that two or more players 
succeed at the same time is at most $\frac{8\e n\ell}{K}$.
\end{restatable}

Finally, as in the two-party case, we show that the unique coinciding equilibrium of a stingy race is also an approximate Nash equilibrium 
in the corresponding quantum race, provided the sequence of probabilities is $\ell$-dense.
\begin{restatable}{theorem}{mapprox}
\label{thm:mapprox}
Let $P_1,....,P_K$ define an $\ell$-dense stingy $n$-player quantum race, $n \ge 2$,  with $4\e n\ell \leq K$. 
If $x = (x_1,...,x_n)$ is the coinciding Nash equilibrium for this stingy race, then $x$ is an $\frac{8\e \ell}{K}$-
approximate Nash equilibrium of the corresponding quantum race.   
\end{restatable}

\subsection{Application to Bitcoin}
One application of our study of quantum races is to the decentralized digital currency Bitcoin, developed in 2008 by 
Satoshi Nakamoto \cite{Nak09}.  Bitcoin transactions are packaged into blocks and stored in a public ledger called the blockchain.  A 
major obstacle in creating a decentralized currency is to find a way for all parties to agree on the history of transactions.  
In Bitcoin, this is done through \emph{Nakamoto consensus}: the right to create a new block is decided through proof-of-work, a contest to 
solve a computational problem.  The winner of this contest has the right to make a new block of 
transactions, is given a reward in bitcoin, and then the process repeats itself.  The players competing in this process are called miners.  Nakamato consensus remains the 
primary means of achieving consensus across all cryptocurrencies, although there are coins using other consensus mechanisms such as proof-of-stake \cite{KRBO17} or Byzantine agreement \cite{CL99}.  

The proof-of-work task used in Bitcoin (originally developed in a system called Hashcash \cite{back02}) is essentially a search 
problem.  The problem is to find a
value $x$ (called a nonce) such that $h(H \parallel x) \le t$, where $h$ is a hash function (doubly iterated SHA-256 in the case of 
Bitcoin), $H$ is the header of the block of transactions to be processed, and $t$ is a hardness parameter that can be varied so 
that the entire network takes 10 minutes to solve this task, on average. 

Several works have studied the impact that quantum computers would have on the Bitcoin protocol \cite{But13, ABLST17,Sat18}, 
both on the mining process we have described above and on the digital signatures used in Bitcoin to authenticate ownership of coins. 
We will focus here on the impact of quantum computers on Bitcoin mining.

As the Bitcoin proof-of-work is a search task, quantum miners could use Grover's algorithm to find a nonce $x$ satisfying 
$h(H \parallel x) \le t$ with quadratically fewer evaluations of the hash function $h$ than is needed by a classical computer \footnote{While this 
seems to give quantum computers a huge advantage for Bitcoin mining, specialized classical Bitcoin mining hardware currently can perform14 
trillion hashes per second \cite{bitmain} and would outperform a near-term quantum computer with gate speeds of 100MHz \cite{ABLST17}}.
The use of Grover's algorithm creates new issues for proof-of-work that do not exist in the classical case.  Desirable properties of a proof-of-work task have been studied 
from an axiomatic point of view by Biryukov and Khovratovich \cite{BK17}.  One property they give is \emph{progress-freeness}: the probability of a miner 
solving the proof-of work task in any moment is independent of 
previous events.  This is achieved for a classical miner in the Bitcoin proof of work, as every call to the hash function is equally 
likely to find a good nonce $x$.  
Progress-freeness is not achieved for a quantum miner running Grover's algorithm, as the success probability grows roughly quadratically 
with the amount of time the algorithm is run.  

Sattath \cite{Sat18} points out that this gives a way for quantum miners to deviate from the prescribed protocol in order to increase their 
chance of winning a block.  To explain this deviation, imagine a simplified case where the proof-of-work is to find a unique marked item in 
a database of size $N$.  Say that Alice, a quantum miner, receives a new block from the network which was found by Bob.  When Alice 
receives this block she will be in the middle of running Grover's algorithm to find the marked item herself.  The prescribed protocol says that she should immediately 
halt this run of Grover's algorithm and begin working on a new search problem by mining on top of Bob's new block.  However, if Alice has run Grover algorithm for $c\sqrt{N}$ steps, 
for a constant $c$, she will have already built up a constant probability of finding the marked item upon measuring.  From the point of view of maximizing 
her payoff, there is no harm in just measuring to see if she finds the marked item.  
If Alice gets lucky and indeed finds the marked item, then she can broadcast her new block to the network.  Depending on her connectivity to the network, some 
other miners may receive Alice's block before Bob's, and there is some probability that Alice's new block eventually becomes the block accepted by the network rather 
than Bob's, 
meaning that Alice will receive the bitcoin reward.  Note that this does not happen in the classical case, where after Alice receives Bob's block she would still 
just have probability $1/N$ to find the marked item with each additional search query.  In this case it makes sense to immediately start mining on top of the new 
block.

Luckily, Sattath also provides an easy fix for the Bitcoin protocol to remedy this problem.  Without going into the technical details, this fix essentially forces 
miners to commit to how long they will run Grover's algorithm \emph{before they begin}.  Thus if Alice commits to running Grover's algorithm for time $\sqrt{N}/100$, 
yet receives Bob's block after time $\sqrt{N}/200$, if she tries to immediately measure and publish her own block, the network will reject it because of the timing 
discrepancy.  This fix fits in very well with our model of quantum races, as a strategy is exactly a probability distribution over choices of times to measure.  

The quantum race that captures the case of Bitcoin mining is what we call the \emph{Grover race} (see \cref{def:grover_race}).  In this race, the 
success probability $p_t$  is given by the success probability of $t$-iterations of Grover's algorithm \footnote{It is known that $t$ queries of Grover's algorithm maximizes 
the probability of success in a search problem over all $t$-query quantum algorithms \cite{DH08}.}.  This race is an $\ell$-dense race for $\ell = \pi /2$.  The size of the 
search space, and thus the maximum number of iterations $K$ to run Grover's algorithm, is determined by the \emph{difficulty} setting of the Bitcoin protocol.  Currently the 
difficulty (as of September 7, 2018) is approximately $7\cdot 10^{12}$, which, by Bitcoin's definition of difficulty, means that the network has to do roughly 
$2^{32} \cdot 7\cdot 10^{12}$ many hashes to succeed, in expectation.  This leads to a value of $K$ of approximately $10^{11}$.  Thus for this application 
$\frac{\ell}{K}$ is very small, and \cref{thm:mapprox} implies that the unique coinciding equilibrium for the stingy Grover race is an $\epsilon$-approximate Nash equilibrium 
in the Grover race for $\epsilon \le 3 \cdot 10^{-10}$.  This gives a reasonable answer to \cref{question1} for what a good strategy would be for quantum miners, and 
moreover has the desirable property that all miners run the same algorithm.  By \cref{colprobm}, when there are $n$ miners running this strategy the probability of a tie is at most 
$3 n \cdot 10^{-10}$.  This gives an answer to \cref{question2}, that quantum mining is not likely to produce a high forking rate and thereby destabilize the Bitcoin protocol.   

\section{Preliminaries}
We use $\e \approx 2.71828$ for Euler's number.
For ease of reference, we state some simple inequalities we will make use of here:
\begin{align}
\label{eq:geom}
\frac{1}{1-x} &\le 1 + 2x \mbox{ for } 0 \le x \le \frac{1}{2} \enspace , \\
\label{eq:recip}
\frac{1}{a+x} &\ge \frac{1}{a} - \frac{x}{a^2} \mbox{ for } a > 0 \mbox{ and } x > -a \enspace , \\
\label{eq:plus_sqrt}
\sqrt{a+x} &\le \sqrt{a} + \frac{x}{2\sqrt{a}} \mbox{ for } a > 0 \mbox{ and } x \ge -a \enspace , \\
\label{eq:minus_sqrt}
\sqrt{a-x} &\le \sqrt{a} - \frac{x}{2\sqrt{a}} - \frac{x^2}{8a^{3/2}} \mbox{  for } a > 0 \mbox{ and } 0 \le x \le a \enspace .
\end{align}

For a probability $0 \leq p \leq 1$, we set $\bar p = 1-p$. 
For a natural number $n$, we let $[n] = \{1, . . . , n\}$.   We let $\Delta_n = \{x \in \mathbb{R}^n : x \ge 0, \sum_{i=1}^n x_i = 1\}$ be the probability simplex.  For 
$x \in \Delta_n$ we let $\sup(x) = \{ i \in [n]: x_i > 0\}$ be the support of $x$.  

By a vector we mean a column vector, so an 
$n$-dimensional vector is an $n$-by-$1$ matrix. We use $\mathbf{0}$ for the all zeros vector and $\mathbf{1}$ for the all ones vector, where 
the dimension is implied from the context.  The transpose of a matrix $M$ is denoted by $M^{\T}$. 
A {\em $2$-player game} 
between players Alice and Bob is specified by two $m \times n$ real matrices
$A$ and $B$. We call $A$ the {\em payoff} matrix of Alice, and $B$ the payoff matrix of Bob. We sometimes refer to
Alice as the {\em row player} and to Bob as the {\em column player}. A {\em pure strategy} of Alice is an element
of $[m]$, and a pure strategy of Bob is an element of $[n]$. We can think of pure strategies of the players as 
their choice of a row and of a column of an $m \times n $ matrix, respectively. When Alice chooses $i$ and Bob 
chooses $j$, their respective payoffs are $a_{ij}$ and $b_{ij}$. 
A {\em mixed strategy} of the row player is a probability distribution $x$ over $[m]$ written as an $m$-vector. 
Similarly, a mixed strategy of the column player is an $n$-vector $y$ describing a probability distribution over columns.
The {\em support} of a mixed strategy is the set of pure strategies (indices) that have positive probabilities.
We denote the support of a mixed strategy $z$ by $\sup(z)$. When Alice plays the mixed strategy $x$ and Bob
plays the mixed strategy $y$, their expected payoffs are $x^{\T}Ay$ and $x^{\T}By$, respectively. 
When Alice uses the pure strategy $i$ against $y$, her payoff is $e_i^{\T}Ay$, where $e_i$ denotes
the $i^{th}$ standard basis vector.
The purpose of both players in the game is to maximize their respective expected payoffs.

A pure strategy $i$ of Alice is a {\em best response} to the mixed strategy $y$ of Bob if it maximizes 
her payoff against $y$.
The set of best responses of Alice against $y$ is denoted by $\br(y)$.
A mixed strategy $x$ of Alice is a {\em best response} to the mixed strategy $y$ of Bob if it contains only best responses, that is $\sup(x) \subseteq \br(y)$. Analogous definitions hold with the roles of the players exchanged.
A pair of mixed strategies $(x,y)$ is a {\em Nash equilibrium} if they are best responses to each other. It is quite
easy to see that $(x,y)$ is a Nash equilibrium exactly when the strategies $x$ and $y$ maximize,
for both players, their expected payoff against the strategy of the other player.

\begin{definition}[quantum race]
A 2-{\em player quantum race} is specified by two sequences of increasing probabilities
$0  < p_1 < p_2 < \cdots < p_K \leq 1$, and $0 < P_1 < P_2 < \cdots < P_K \le 1$ for some integer $K \ge 2$. The set of pure strategies
of both players is $[K]$. The $K \times K$ payoff matrix $A$ of Alice is defined as
$$
A(i,j) =
\begin{cases}
p_i  & \text{if} ~ i < j,\\
p_i \bar P_j + \frac{1}{2}p_iP_i & \text{if} ~ i = j, \\
p_i \bar P_j & \text{otherwise.}
\end{cases}
$$
The payoff matrix $B$ is defined as 
$$
B(i,j) =
\begin{cases}
P_j  & \text{if} ~ j < i,\\
P_j \bar p_i + \frac{1}{2} p_i P_i& \text{if} ~ i = j, \\
P_j \bar p_i & \text{otherwise.}
\end{cases}
$$
If $p_i = P_i$ for all $i=1, \ldots, K$ then we call the game a \emph{symmetric} quantum race.  Note that in this case 
$B=A^{\T}$.
\end{definition}

In our study of quantum races, a key role will be played by an auxiliary game that is easier to analyze called a stingy quantum race.  
A stingy quantum race differs from a quantum race only in that no payout is given in the case of a tie.
\begin{definition}[stingy quantum race]
A 2-{\em player stingy quantum race} is specified by two sequences of increasing probabilities
$0  < p_1 < p_2 < \cdots < p_K \leq 1$, and $0 < P_1 < P_2 < \cdots < P_K \le 1$ for some integer $K \ge 2$. The set of pure strategies
of both players is $[K]$. The $K \times K$ payoff matrix $A_0$ of Alice is defined as
$$
A_0(i,j) =
\begin{cases}
p_i  & \text{if} ~ i < j,\\
p_i \bar P_j & \text{otherwise.}
\end{cases}
$$
The payoff matrix of Bob $B_0$ is defined as 
$$
B_0(i,j) =
\begin{cases}
P_j  & \text{if} ~ j < i,\\
P_j \bar p_i & \text{otherwise.}
\end{cases}
$$
If $p_i = P_i$ for all $i=1, \ldots, K$ then we call the game a \emph{symmetric} stingy quantum race.  Note that in this case 
$B_0=A_0^{\T}$.
\end{definition}

The main specific quantum race we will be interested in is the \emph{Grover race}.  This results from two players competing 
to find a marked item in a database and playing by running Grover's algorithm for a certain amount of time and then measuring. 
Formally, the race is defined as follows.  
\begin{definition}[Grover race]
\label{def:grover_race}
We define the (stingy) Grover race on $N$ items as the symmetric (stingy) quantum race with
$K = \ceil{ \tfrac{\pi}{4}\sqrt{N} - 3/2}$ and
\[
p_t = \sin^2 \left(2(t+1/2)\arcsin\left(\frac{1}{\sqrt{N}}\right)\right) \enspace,
\]
for $1 \leq t \leq K$.
\end{definition}
Here $p_t$ is the success probability of Grover's algorithm of finding a unique marked item in a database of 
$N$ items.  It is known that $p_t$ is the highest success probability for finding a marked item 
for any quantum algorithm making $t$ many calls to the database \cite{DH08}. 

The Grover race has many nice properties, and we will abstract out one of them here.  This allows us to show results 
for a general class of quantum races, rather than just the Grover race.
\begin{definition}[dense race]
\label{def:dense}
Let $p_1 < p_2 < \cdots < p_K \le 1$.  We call 
the sequence $(p_1, \ldots, p_K)$ $\ell$-dense if $p_1 \le \frac{\ell}{K}$, $p_K \ge 1-\frac{\ell}{K}$, and 
$p_{i+1} - p_i \le \tfrac{\ell}{K}$ for all $i=1, \ldots, K-1$.  
For a dense sequence $(p_1, \ldots, p_K)$ will similarly call the symmetric (stingy) quantum race defined by this 
sequence a symmetric (stingy) $\ell$-dense quantum race.
\end{definition}

The (stingy) Grover race is an $\ell$-dense race with $\ell= \pi/2$.

\section{Two-player stingy quantum races}
\label{sec:too_stingy}
We will first analyze Nash equilibria in stingy quantum races.  We can show several structural properties about the support of 
Nash equilibria in stingy quantum races that make them easier to analyze than quantum races.  After our analysis in this 
section, we will  bootstrap our knowledge of Nash equilibria in stingy quantum races to find an approximate Nash 
equilibrium for quantum races. 

Consider a stingy quantum race given by the probabilities $0 < p_1 < p_2 < \cdots < p_K \leq 1$ and 
$0 < P_1 < P_2 < \cdots < P_K$,
and let $y$ be a mixed strategy of Bob. For $t \leq K$, we let 
$\sup_{\leq t }(y)= \{j \in \sup(y) : j \le t\}$ be the set of times played by Bob that are at most $t$, and similarly,
we let 
$\sup_{> t }(y)= \{j \in \sup(y) : j > t\}$ be the set of times played by Bob that are greater than $t$.
Observe that when Alice plays the pure strategy $t$ against $y$,
her payoff is 
$$
e_t^{\T}Ay = p_{t} \left(\sum_{j \in \sup_{\leq t}(y)} y_j \bar{P}_j + \sum_{j \in \sup_{> t}(y)} y_j \right).
$$

\begin{claim}
\label{cla:interval}
Let $(x, y)$ be a Nash equilibrium of a 2-player stingy quantum race.
If $t_1 \in \sup (x)$ then there does 
not exist $t_2 > t_1$ with $\sup_{\leq t_1}(y) = \sup_{\leq t_2}(y)$.  
\end{claim}

\begin{proof}
Say that the game is defined by probabilities $0 < p_1 < p_2 < \cdots < p_K \leq 1$ and 
$0 < P_1 < P_2 < \cdots < P_k$
If $\sup_{\leq t_1}(y) = \sup_{\leq t_2}(y)$ then
$$
e_{t_2}^{\T}Ay = \frac{p_{t_2}}{p_{t_1}} e_{t_1}^{\T}Ay.
$$
As $p_{t_1} < p_{t_2}$,  the payoff for playing 
$t_2$ is strictly larger than that for playing $t_1$.  
Therefore $t_1$ is not a best response for $y$, in contradiction with the definition of a Nash equilibrium.

\end{proof}

This claim implies that the support structure of Nash equilibria in stingy quantum races
is relatively simple.  We first make the following definition.

\begin{definition}
A pair of strategies $(x,y)$ is called {\em coinciding} if $\sup(x) = \sup(y)$.
A pair of strategies $(x,y)$ is called {\em alternating} if there exists $1\le t_1 < t_2 \le K$ such that 
the support of one player is
$\{t_1, t_1+2, \ldots, t_2 -1\}$ and the support of the other is $\sup(y)=\{t_1+1,t_1+3, \ldots, t_2\}$.
A pair of strategies $(x,y)$ is called {\em $(t_1,c,t_2)$-alternating-coinciding} if there are 
natural numbers $1 \le t_1 < t_2 \le K$ and $t_1+2 \le c \le t_2$ such that the support of one player is
$\{t_1, t_1 +2, \ldots, c-2, c, c+1, c+2, \ldots, t_2\}$ and the support of the other is
$\sup(y) = \{t_1+1, t_1 +3, \ldots, c-3, c-1,c, c+1, c+2, \ldots, t_2\}$.  
\end{definition}

\begin{corollary}
\label{cor:support_structure}
Let $(x, y)$ be a Nash equilibrium of a stingy quantum race specified by probabilities 
$0 < p_1 < p_2 < \cdots < p_K \leq 1$ and $0 < P_1 < P_2 < \cdots < P_k$. Then
there is some $1 \le T \leq K$ such that 
\begin{itemize}
\item $\sup(x) \cup \sup(y) = [T,K]$ \enspace ,
\item $(x,y)$ is either coinciding, alternating, or alternating-coinciding.
\end{itemize}
\label{cor:interval}
\end{corollary}
\begin{proof}
From \cref{cla:interval} we can easily derive the following two statements:
\begin{itemize}
   \item $\sup(x) \cup \sup(y)$ is an interval containing the time with maximum success probability,
   \item For every $t_1, t_2 \in \sup(x)$ there must be $t \in \sup(y)$ with $t_1 < t \le t_2$.  
\end{itemize}
These statements immediately imply the claim.
\end{proof}

We now study the particularly simple coinciding equilibria.  In \cref{sec:app_alt} we give a full characterization 
of all Nash equilibria in a symmetric stingy quantum race.

\subsection{Unique coinciding equilibrium}
We first look for Nash equilibria where the mixed strategies of Alice and Bob have the same support.  
If the number of strategies is $K$,
we know by 
\cref{cor:interval} that this support must be a set $\{T, T+1, \ldots, K\}$, for some 
$1 \leq T \leq K$.  
\begin{lemma}
\label{lem:characterize}
Consider a stingy quantum race defined by 
$0  < p_1 < \ldots < p_K \leq 1$ and $0 < P_1 < P_2 < \cdots < P_k$.  
Let $x,y \in \mathbb{R}^K$.  Then $(x,y)$ is a Nash equilibrium for this game 
with support $\{T, T+1, \ldots, K\}$, for some $1 \leq T \leq K$,
if and only if $x$ and $y$ satisfy the following system of equations and inequalities.
\begin{align}
e_t^{\T}Ay & =  e_{T}^{\T}Ay, ~  & \text{for} ~~ T < t \leq K ~, \label{eq:ne1}\\
e_{T-1}^{\T}Ay & \leq e_{T}^{\T}Ay,    & \label{eq:ne2} \\
y_t & =   0  & \text{for} ~~ 0 < t < T ~, \label{eq:zero} \\
y_t & > 0  & \text{for} ~~ T \leq t \leq K ~, \label{eq:ne3} \\
\sum_{t = T}^K y_t &= 1 & \label{eq:ne4} \\
x^{\T}Be_t & =  x^{\T}Be_T, ~  & \text{for} ~~ T < t \leq K ~,\label{eq:ne1x} \\\
x^{\T}Be_{T-1} & \leq x^{\T}Be_T,    & \label{eq:ne2x} \\
x_t & =   0  & \text{for} ~~ 0 < t < T ~, \label{eq:zerox}\\
x_t & > 0  & \text{for} ~~ T \leq t \leq K ~, \\
\sum_{t = T}^K x_t &= 1~ . & \label{eq:ne4x}
\end{align}
\end{lemma}
\begin{proof}
\cref{eq:zero}--(\ref{eq:ne4}) and (\ref{eq:zerox})--(\ref{eq:ne4x}) express that $x$ and $y$ are probability distributions with 
support $\{T, T+1, \ldots, K\}$.  The other conditions for a Nash equilibrium are that all strategies in the support of $x$ are best 
responses against $y$ and vice versa.  That all strategies in the support of $x$ are best responses against $y$ means
\begin{align*}
e_t^{\T}Ay & =  e_{T}^{\T}Ay, ~  & \text{for} ~~ T < t \leq K ~, \\
e_{t-1}^{\T}Ay & \leq e_{t}^{\T}Ay,    &  \text{for} ~~ 1 \le t < T \\
\end{align*}
The first equation here is exactly \cref{eq:ne1}.  The inequality here is implied by \cref{eq:ne2}.  This is because 
for $t < T-1$
\[ 
e_{t}^{\T}Ay = \frac{p_{t}}{p_{T-1}} e_{T-1}^{\T} A y 
\]
as $\sup_{\le T-1}(y) = \sup_{\le t}(y)$.  A similar argument show that \cref{eq:ne1x}--(\ref{eq:ne2x}) show that all strategies 
in the support of $y$ are best responses against $x$.
\end{proof}

When $P_K=1$, then Alice has zero payoff on playing time $K$.  Thus as long as $K \ge 2$, 
when $P_K=1$ there is no coinciding Nash equilibrium $(x,y)$ where $\sup(x)=\sup(y)=\{K\}$.  A similar 
argument applies when $p_K=1$.  We will therefore 
exclude the case $T=K$ and either $p_K=1$ or $P_K=1$ for the next definition and \cref{lem:unicity}.

\begin{definition}
We define the values $q_T^A,q_T^B$, for $T = 2, \ldots , K$, and $r_T^A,r_T^B,z_T^A,z_T^B$, for $T = 1, \ldots , K-1$.  The 
values $z_K^A, r_K^A$ are not defined when $T=K, p_K=1$ and $z_K^B, r_K^B$ are not defined when 
$T=K, P_K=1$.  
\begin{align}
q_i^A &= \frac{1}{p_i}\left(\frac{1}{P_{i-1}} - \frac{1}{P_i} \right)  \\
q_i^B &= \frac{1}{P_i}\left(\frac{1}{p_{i-1}} - \frac{1}{p_i} \right) \\
r_T^A &= \frac{1}{\bar{p}_T}\left(\frac{1}{P_K} - \sum_{i=T+1}^K \bar{p}_i q_i^A\right) \\
r_T^B &= \frac{1}{\bar{P}_T}\left(\frac{1}{p_K} - \sum_{i=T+1}^K \bar{P}_i q_i^B\right)  \\
z_T^A &= r_T^A + \sum_{i = T+1}^K q_i^A \\
z_T^B &= r_T^B + \sum_{i = T+1}^K q_i^B 
\end{align}
\end{definition}

\begin{lemma}
\label{lem:unicity}
For every $1\leq T \leq K-1$, and the case $T=K$ and $P_K \ne 1$,
the system of linear equations composed of the \cref{eq:ne1},~(\ref{eq:zero}) and~(\ref{eq:ne4})
of \cref{lem:characterize} has a unique solution given by
\[
y_t =
\begin{cases}
{r_{T}^B}/{z_{T}^B}& \text{if} ~~~ t= T,\\
{q_t^B}/{z_{T}^B}& \text{if} ~~~ T < t \leq K,\\
0 & \text{otherwise} \enspace .
\end{cases}
\]
\end{lemma}

\begin{proof}
For convenience, we drop the normalization condition $\sum_{t = T}^K y_t  = 1$ and 
instead scale $y$ such that the expected payoff of a best response is 1.  That is, we replace \cref{eq:ne4} by
\cref{eq:ne5}:
\begin{equation}
\label{eq:ne5}
e_{T}^{\T}Ay =1.
\end{equation}
Clearly the solutions of \cref{eq:ne1},(\ref{eq:zero}),(\ref{eq:ne4}) and the solutions of 
\cref{eq:ne1},(\ref{eq:zero}),(\ref{eq:ne5}) differ only by a constant multiplicative factor, and from a solution of the latter system
a solution of the former one can be obtained by dividing it coordinate-wise by $\sum_{t = T}^K y_t$.

We first want to find the solutions of \cref{eq:ne1} and~(\ref{eq:ne5}). Together, they can be expressed in matrix form
as:

\begin{equation}
\label{eq:main}
\begin{bmatrix}
p_T\bar{P}_T & p_T & p_T & \cdots & p_T \\
p_{T+1}\bar{P}_T & p_{T+1} \bar{P}_{T+1} & p_{T+1} & \cdots & p_{T+1} \\
p_{T+2}\bar{P}_T & p_{T+2}\bar{P}_{T+1} & p_{T+2}\bar{P}_{T+2} & \cdots & p_{T+2} \\
\vdots & & &\ddots & \vdots \\
p_K\bar{P}_T & p_K \bar{p}_{T+1} & p_K\bar{P}_{T+2} & \cdots & p_K\bar{P}_K \\
\end{bmatrix}
\begin{bmatrix}
y_T \\ y_{T+1} \\ y_{T+2} \\ \vdots \\ y_K
\end{bmatrix}
=
\begin{bmatrix}
1 \\ 1 \\ 1 \\ \vdots \\ 1
\end{bmatrix}.
\end{equation}
We can simplify the above system of linear equations as follows:
\[
\label{eq:main_simp}
\begin{bmatrix}
\bar{P}_T & 1 & 1 & \cdots & 1 \\
\bar{P}_T & \bar{P}_{T+1} & 1 & \cdots & 1 \\
\bar{P}_T & \bar{P}_{T+1} & \bar{P}_{T+2} & \cdots & 1 \\
\vdots & & & \ddots & \vdots \\
\bar{P}_T & \bar{P}_{T+1} & \bar{P}_{T+2} & \cdots & \bar{P}_K \\
\end{bmatrix}
\begin{bmatrix}
y_T \\ y_{T+1} \\ y_{T+2} \\ \vdots \\ y_K
\end{bmatrix}
=
\begin{bmatrix}
1/p_T \\ 1/p_{T+1} \\ 1/p_{T+2} \\ \vdots \\ 1/p_K
\end{bmatrix}.
\]
Now for each row before the last we subtract from it the next row, which gives:
\[
\begin{bmatrix}
0 & P_{T+1} & 0 & \cdots & 0 \\
0 & 0 & P_{T+2} & \cdots & 0 \\
0 & 0 & 0 & \cdots & 0 \\
\vdots & & & \ddots & \vdots \\
\bar{P}_T & \bar{P}_{T+1} & \bar{P}_{T+2} & \cdots &\bar{P}_K \\
\end{bmatrix}
\begin{bmatrix}
y_T \\ y_{T+1} \\ y_{T+2} \\ \vdots \\ y_K
\end{bmatrix}
=
\begin{bmatrix}
1/p_T-1/p_{T+1} \\ 1/p_{T+1}-1/p_{T+2} \\ 1/p_{T+2} - 1/p_{T+3} \\ \vdots \\ 1/p_K
\end{bmatrix}.
\]
From this we can see that, for every $1 \leq T \leq K$, \cref{eq:main} has
a unique solution, given by $y_T = r_T^B$, and $y_t = q_t^B$, for $T < t \le K$.
\end{proof}
\begin{definition}
\label{def:tstar}
Define $T^*_A$ (respectively $T^*_B$) as the smallest integer $1 \leq T \leq K$ such that $r_T^A > 0$ 
(respectively $r_T^B > 0$).  When $T^*_A = T^*_B$ we denote their common value as $T^*$.  
\end{definition}
This definition makes sense as in the case $p_K=1$ (when $r_K^A$ is not defined) we see that $r_{K-1}^A > 0$ and otherwise 
$r_K^A = \tfrac{1}{\bar{p}_K P_K} >0$.  A similar argument applies for $r_K^B$.

The following theorem characterizes the coinciding Nash equilibria in a stingy quantum race.
\begin{lemma}
\label{thm:unique1-5}
\cref{eq:ne1}--(\ref{eq:ne4}) have a solution if and only if $T=T^*_B$.  In the case $T=T^*_B$, the solution is 
unique and is given by 
$$
y_{t} =
\begin{cases}
{r_{T^*_B}^B}/{z_{T^*_B}^B}& \text{if} ~~~ t= T^*_B,\\
{q_t^B}/{z_{T^*_B}^B}& \text{if} ~~~ T^*_B < t \leq K. 
\end{cases}
$$
\end{lemma}

\begin{proof}
If $K \ge 2$ and $P_K = 1$ then $e_K^TAe_K =0$ and $e_{K-1}^TA e_K > 0$.  Thus \cref{eq:ne1}--(\ref{eq:ne4}) 
do not have a solution with $\sup(y) = \{K\}$ in this case.  We therefore exclude the case $T=K$ and $p_K=1$ from the discussion 
below.  

By \cref{lem:unicity} we know that for every $1 \leq T \leq K$, there exists a unique solution to 
\cref{eq:ne1},~(\ref{eq:zero}) and~(\ref{eq:ne4}). 
For a fixed $T$, the only possible solution is given by \cref{lem:unicity}.
We now examine for which $T$ is it true that the this solution $y$
also satisfies \cref{eq:ne2} and~(\ref{eq:ne3}).
First we claim that \cref{eq:ne3} is satisfied if and only if $T \geq T^*_B$.
To see that, observe that
as the $p_t$ form an increasing sequence, we have $q_t^B> 0$, for all $2 \leq t \le K$.  
Therefore \cref{eq:ne3} is satisfied if and only if $r_T^B > 0$ which, because 
$r_T^B$ is an increasing function of $T$, holds if and only if $T \geq T^*_B$.

We now turn to \cref{eq:ne2}, and we claim that it is satisfied if and only if 
$T \leq T^*_B$.  Substituting in the values of $y$ from \cref{lem:unicity}, we can express 
\cref{eq:ne2} as
\begin{equation}
\label{eq:less}
p_{T-1} \left(  r_{T}^B   +  \sum_{i = T+1}^K  q_i^B \right) \leq 1 .
\end{equation}
Let us also observe that the first row of \cref{eq:main} applied with the start of the interval at $T - 1$ gives
\begin{equation}
\label{eq:equal}
p_{T-1} \left( \bar{P}_{T-1} r_{T-1}^B   +  \sum_{i = T}^K  q_i^B \right) = 1 ,
\end{equation}
for $T = 2, \ldots, K$.

Comparing \cref{eq:less} and~(\ref{eq:equal}), we see that \cref{eq:ne2} holds if and only if
\begin{equation}
\label{eq:final}
r_T^B \leq  \bar{P}_{T-1} r_{T-1}^B   + q_T^B.
\end{equation}
Note that for every $2 \leq T \leq K$, we have
\begin{equation}
\label{eq:yT}
\bar{P}_{T-1} r_{T-1}^B = {\bar{P}_T} (r_T^B - q_T^B) \enspace.
\end{equation}
Replacing $\bar{P}_{T-1} r_{T-1}^B$ by $ {\bar{P}_T} (r_T^B - q_T^B)$ in the right hand side of \cref{eq:final}, we get
that \cref{eq:ne2} is equivalent to
\begin{equation}
\label{eq:finall}
r_T^B \leq  \bar{P}_{T} r_{T}^B   + P_Tq_T^B.
\end{equation}
The right hand side of \cref{eq:finall} is a convex combination of $r_T^B$ and $q_T^B$, therefore it is at least $r_T$
if and only if 
\begin{equation}
\label{eq:finalll}
r_T^B - q_T^B \leq 0.
\end{equation}
Looking again at \cref{eq:yT}, since both $\bar{P}_{T-1}$ and $\bar{P}_{T}$ are positive, we can deduce that 
\cref{eq:finalll}, and therefore also \cref{eq:ne2} holds exactly when
\begin{equation}
\label{eq:lala}
r_{T-1}^B \leq  0.
\end{equation}
By the definition of $T^*_B$, and the fact that $r_T^B$ is an increasing function of $T$, 
\cref{eq:lala} holds when $T \leq T^*_B$.  Thus we conclude that \cref{eq:ne1}--(\ref{eq:ne4}) are 
satisfied if and only if $T=T^*_B$, in which case the solution is given as in \cref{lem:unicity}.
\end{proof}

The following theorem characterizes the coinciding Nash equilibria in a stingy quantum race.
\begin{theorem}
\label{thm:symmetric}
A stingy quantum race defined by probabilities 
$0  < p_1 < \ldots < p_K \leq 1$ and $0  < P_1 < \ldots < P_K \leq 1$ has a coinciding Nash equilibrium 
if and only if $T_A^* = T_B^*$. In this case, letting $T^* = T^*_A = T^*_B$ there is a unique coinciding equilibrium given by 
\[
x_{t} =
\begin{cases}
{r_{T^*}^A}/{z_{T^*}^A}& \text{if} ~~~ t= T^*,\\
{q_t^A}/{z_{T^*}^A}& \text{if} ~~~ T^* < t \leq K
\end{cases} , \qquad
y_{t} =
\begin{cases}
{r_{T^*}^B}/{z_{T^*}^B}& \text{if} ~~~ t= T^*,\\
{q_t^B}/{z_{T^*}^B}& \text{if} ~~~ T^* < t \leq K \enspace. 
\end{cases} 
\]
In particular, when $p_i = P_i$ for all $1 \le i \le K$ then $(x,x)$ is the unique coinciding Nash equilibrium.
\end{theorem}

\begin{proof}
By \cref{lem:characterize} a coinciding Nash equilibrium $(x,y)$ supported on $\{T, T+1, \ldots, K\}$ must 
satisfy \cref{eq:ne1}--(\ref{eq:ne4x}).  By \cref{thm:unique1-5}, \cref{eq:ne1}--(\ref{eq:ne4}) 
are satisfied
if and only if $T = T^*_B$ and $y$ is given as in the Lemma.  We can also apply \cref{thm:unique1-5} to 
(the transpose of) \cref{eq:ne1x}--(\ref{eq:ne4x}) to see that they have a solution if and only if 
$T = T^*_A$ and $x$ is given by 
$$
x_t =
\begin{cases}
{r_{T^*_A}^A}/{z_{T^*_A}^A}& \text{if} ~~~ t= T^*_A,\\
{q_t^A}/{z_{T^*_A}^A}& \text{if} ~~~ T^*_A < t \leq K,\\
0 & \text{otherwise}.
\end{cases}
$$
As $x$ and $y$ must have the same support in a coinciding Nash equilibrium, there can only exist a coinciding 
Nash equilibrium if $T^*_A = T^*_B$.  

When $p_i = P_i$ for all $1 \le i \le K$ then clearly $r_T^A = r_T^B$ and $q_i^A = q_i^B$ and it will always 
be the case that $T^*_A = T^*_B$.  Thus there will always exist a Nash equilibrium in this case, given by the 
unique solution to \cref{eq:ne1}--(\ref{eq:ne4x}).
\end{proof}

\subsection{Payoff and collision probability}
In this section we will explore the consequences of the coinciding Nash equilibrium we have found for 
the payoff of the game and for the probability that the two players win at the same time, the collision probability.
For these results, we will only consider the symmetric case when
there is always a unique symmetric equilibrium whose support begins at $T^*=T^*_A = T^*_B$.  Note that the payoff 
for each player with this strategy is $\frac{1}{z_{T^*}}$.  Since a player receives payoff $1$ upon winning, $\frac{1}{z_{T^*}}$ 
is also exactly the each player's winning probability.  

To investigate the collision probability, we will also make the following definitions.
\begin{definition}[Unnormalized collision probability]
Define
\[
\sigma(T) = p_T^2 r_T^2 + \sum_{i=T+1}^K p_i^2 q_i^2 \enspace.
\]
\end{definition}
With this definition, $\frac{1}{z_{T^*}^2}\sigma(T^*)$ is the collision probability we are interested in.

First we analyze the payoff in a symmetric stingy quantum race.
\begin{theorem}
\label{thm:tstar_upper}
Let $p_1 < p_2 < \cdots < p_K$ define a stingy symmetric quantum race.  Then 
\[
z_{T^*}=1 + \sqrt{1+\tfrac{1}{p_K^2}+\sigma(T^*)} \enspace .
\]
In particular, 
\[
\frac{1}{z_{T^*}} \le \sqrt{2}-1 \enspace .
\]
\end{theorem}
\begin{proof}
We look at all different possible outcomes of the race: Alice wins, Bob wins, or no one wins.  There are two 
different ways in which no one wins: Alice and Bob both succeed at the same time, or, no one 
succeeds.  The probability that Alice wins is $1/z_{T^*}$, and the same for Bob.  Thus we have
\begin{align*}
1 &= \frac{2}{z_{T^*}} + \frac{\sigma(T^*)}{z_{T^*}^2} + 
\left(\frac{1}{z_T^*} \left( \bar{p}_{T^*}r_{T^*}+\sum_{i=T^*+1}^K \bar{p}_i q_i \right)\right)^2 \\
&= \frac{2}{z_{T^*}} + \frac{\sigma(T^*)}{z_{T^*}^2} + \left(\frac{1}{p_Kz_{T^*}} \right)^2 \enspace .
\end{align*}
Taking the positive root gives
\[
z_{T^*} = 1 + \sqrt{1+\tfrac{1}{p_K^2}+\sigma(T^*)} \enspace .
\]
The ``in particular'' statement follows by noting $z_{T^*} \ge 1 + \sqrt{2}$.  
\end{proof}

\begin{corollary}
\label{cor:ptstar}
If $T^* \ge 2$ then
\[
p_{T^*-1} \le \sqrt{2}-1 \enspace .
\]
\end{corollary}

\begin{proof}
As can be seen from Bob playing time $T^*-1$, we 
have $p_{T^*-1} z_{T^*} \le 1$, thus $p_{T^*-1} \le \sqrt{2}-1$ by \cref{thm:tstar_upper}.
\end{proof}

Although \cref{thm:tstar_upper} gives an exact expression for the payoff, we would like to get a general lower bound 
on the payoff.  This requires showing an upper bound on the collision probability.  Showing an upper bound on the collision 
probability is also important for the application to Bitcoin, to estimate the forking probability amongst quantum miners.

The first step to upper bounding the collision probability is to get a rough lower bound on $p_{T^*}$.  This is our initial 
bootstrap, which will then let us upper bound the collision probability and then in turn get a sharper lower bound on $p_{T^*}$ 
in \cref{cor:payoff}.  For these results we restrict to $\ell$-dense stingy quantum races.

\begin{lemma}
\label{p_T_lb_2}
Let $p_1, \ldots, p_K$ define an $\ell$-dense symmetric stingy quantum race.  If $K \ge 6\ell$ then 
$p_{T^*} > \tfrac{5}{21}$.  In particular, $T^* \ge 2$.  
\end{lemma}

\begin{proof}
As $r_{T^*} > 0$ we have 
\begin{align*}
\frac{1}{p_K} &> \sum_{i=T^*+1}^K \bar{p}_i q_i \\
&= \sum_{i=T^*+1}^K \left( \frac{1}{p_i} -1 \right) \left( \frac{1}{p_{i-1}} - \frac{1}{p_i} \right)
\end{align*}
Let $L = \max_T \ \{T: p_T \le 1/2\}$.  By assumption of an $\ell$-dense quantum race, we have
 $p_L \ge \tfrac{1}{2}-\tfrac{\ell}{K}$.  Whenever $p_i \le 1/2$ we have $(1/p_i -1) \ge 1$, thus
\begin{align*}
\frac{1}{p_K} &> \sum_{i=T^*+1}^L \frac{1}{p_{i-1}} - \frac{1}{p_i} \\
&= \frac{1}{p_{T^*}} - \frac{1}{p_L} \\
&\ge \frac{1}{p_{T^*}} -3
\end{align*}
where for the last inequality we have used $\tfrac{1}{p_L} \le 3$ by the assumption $K \ge 6\ell$.  
This then implies that $\tfrac{1}{p_{T^*}} \le \tfrac{6}{5}+3$, again using the fact that $K \ge 6\ell$ and 
so $p_K \ge \tfrac{5}{6}$.  

The ``in particular'' holds since $p_1 \le 1/6$ by the assumption of $\ell$-density and that $K \ge 6 \ell$, 
and $1/6 < 5/21$.  
\end{proof}  

\begin{theorem}
\label{thm:col}
Let $p_1, \ldots, p_K$ define an $\ell$-dense symmetric stingy quantum race.  If $K \ge 6 \ell$ then 
\[
\frac{\sigma(T^*)}{z_{T^*}^2} \le \frac{6\ell}{K} \qquad \mbox{ and } \qquad \sigma(T^*) \le \frac{196\ell}{K} \enspace .
\]
\end{theorem}

\begin{proof}
We will actually prove a more general statement:
for any probability distribution $v$ on $[K]$ it holds that 
\begin{equation}
\label{eq:more_general}
\frac{1}{z_{T^*}} \left( v_{T^*} p_{T^*}^2 r_{T^*} + \sum_{i=T^*+1}^K v_i p_i^2 q_i \right) \le \frac{6 \ell}{K} \enspace .
\end{equation}
Applying this with the distribution $v_{T^*} = \frac{r_{T^*}}{z_{T^*}}$ and $v_T = \frac{q_T}{z_{T^*}}$ for $T^*+1 \le T \le K$ gives 
the first inequality.  We discuss how to derive the second inequality of the Theorem after the proof of \cref{eq:more_general}. 

Now we turn to the proof of \cref{eq:more_general}.  First, note that $T^* \ge 2$ by \cref{p_T_lb_2} and thus it makes sense to talk about $T^*-1$.
As $r_{T^*} \le q_{T^*}$ we have:
\begin{align*}
\frac{1}{z_{T^*}} \left( v_{T^*} p_{T^*}^2 r_{T^*} + \sum_{i=T^*+1}^K v_i p_i^2 q_i \right) 
&\le \frac{1}{z_{T^*}} \sum_{i=T^*}^K v_i p_i^2 q_i \\
&= \frac{1}{z_{T^*}} \sum_{i=T^*}^K v_i p_i \left(\frac{1}{p_{i-1}} - \frac{1}{p_i} \right) \\
&= \frac{1}{z_{T^*}} \sum_{i=T^*}^K v_i \left(\frac{p_i - p_{i-1}}{p_{i-1}}\right) \\
&\le \frac{\ell}{Kz_{T^*}} \sum_{i=T^*}^K v_i \left(\frac{1}{p_{i-1}}\right) \\
&\le \frac{\ell}{Kz_{T^*}} \frac{1}{p_{T^*-1}} \enspace ,
\end{align*}
since $v$ is a probability distribution.  Now note that 
$p_{T^*-1} \ge \frac{5}{21}-\frac{1}{6} \ge \frac{1}{14}$ by \cref{p_T_lb_2} and the assumption $\frac{\ell}{K} \le \frac{1}{6}$.
\cref{eq:more_general} now follows from the upper bound $\frac{1}{z_{T^*}} \leq \sqrt{2}-1$ from 
\cref{thm:tstar_upper}.  

To see the second inequality, note 
\[
\sigma(T^*) \le \frac{\ell}{K} \frac{z_{T^*}}{p_{T^*-1}} \le \frac{\ell}{K} \frac{1}{p_{T^*-1}^2} \le \frac{196\ell}{K} \enspace .
\]

\end{proof}

Now that we have an upper bound on the collision probability, we obtain the following corollary to \cref{thm:tstar_upper}.
\begin{corollary}
\label{cor:payoff}
Let $p_1, \ldots, p_K$ define an $\ell$-dense symmetric stingy quantum race.  Let 
$\tau = \frac{50\sqrt{2}\ell}{K}$.  If $K \ge 6 \ell$ then 
\[
z_{T^*} \le \sqrt{2}+1 + \tau, \qquad \frac{1}{z_{T^*}} \ge  \sqrt{2}-1 - 
\tau(\sqrt{2}-1)^2, \qquad p_{T^*} \ge \sqrt{2}-1 - \tau(\sqrt{2}-1)^2 \enspace .
\]
\end{corollary}

\begin{proof}
\label{cor:col}
By \cref{thm:tstar_upper} we have
\begin{align*}
z_{T^*} &= 1 + \sqrt{1 + \frac{1}{p_K^2} + \sigma(T^*)} \\
&\le 1 + \sqrt{1 + \frac{K^2}{(K-\ell)^2} + \frac{196\ell}{K}} \\
&\le 1 + \sqrt{2 + \frac{200\ell}{K}} \enspace,
\end{align*}
as $\frac{K^2}{(K-\ell)^2} \le 1 + \frac{4\ell}{K}$ by \cref{eq:geom} (using the assumption that $\frac{\ell}{K} \le \frac{1}{6}$).
Continuing by applying \cref{eq:plus_sqrt} we have 
\[
z_{T^*} \le \sqrt{2} + 1 + \frac{50\sqrt{2} \ell}{K} \enspace, 
\]
giving the first inequality of the Corollary.
The second inequality of the Corollary follows by applying \cref{eq:recip} to the first inequality.

The third inequality of the Corollary follows from the second inequality from the observation that $p_{T^*} \ge \frac{1}{z_{T^*}}$.  
This is because 
\begin{align*}
\frac{1}{z_{T^*}} &= \frac{p_{T^*}}{z_{T^*}}\left(r_{T^*}\bar{p}_{T^*} + \sum_{t = T^*+1}^K q_{t} \right) \\
&\le \frac{p_{T^*}}{z_{T^*}}\left(r_{T^*}+ \sum_{t = T^*+1}^K q_{t} \right) \\
&=p_{T^*} \, .
\end{align*}
\end{proof}

\section{Two-player quantum races}
\label{sec:two_races}
In this section, we bootstrap our results about symmetric stingy quantum races to analyze symmetric quantum races.  
Our main results are two-fold.  
\begin{enumerate}
\item The unique coinciding Nash equilibrium in an $\ell$-dense  symmetric stingy quantum race is an approximate 
Nash equilibrium in the corresponding quantum race.
\item The approximate Nash equilibrium from~(1) achieves a payoff which is nearly optimal among all symmetric 
Nash equilibria in a symmetric quantum race.
\end{enumerate}

The intuition for item~(1) is clear.  The only difference between a stingy quantum race and a quantum race is the payoff on ties.  
For the unique coinciding Nash equilibrium we have shown that the collision probability is
$O(\ell/K)$, thus the change in payoff on ties 
will make only a small change to the payoffs under this strategy.  

For item~(2), we use the quadratic programming characterization of Nash equilibria \cite{MS1964}.  Consider a game $(A,B)$ 
where $A,B$ are $m$-by-$n$ matrices.  The program
\begin{equation}
\label{eq:quad_prog_gen}
\begin{aligned}
& \underset{x \in \Delta^m, y \in \Delta^n, \alpha,\beta \in \R}{\text{maximize}}
& &  x^T (A+B) y -\alpha - \beta \\
& \text{subject to}
& & A y \le \alpha \mathbf{1},\\
&& & B^T x \le \beta \mathbf{1},\\
\end{aligned}
\end{equation}
has an optimal value of $0$, and any $(x,y)$ attaining the value $0$ is a Nash equilibrium.  In the case of a symmetric quantum 
race $(A,A^T)$, when we restrict to symmetric strategies $(x,x)$ the objective function in \cref{eq:quad_prog_gen} becomes 
negative definite plus linear, making this a standard quadratic program.  
We then use the tight dual formulation of a quadratic program \cite{Dorn1960} to give an upper bound on the 
payoff of any symmetric Nash equilibrium, by explicitly constructing solutions to the dual problem.  This allows 
us to show that the payoff of the unique coinciding equilibrium in a stingy race achieves a payoff which is within 
$O(\sqrt{\ell/K})$ of optimal amongst all symmetric equilibria in the corresponding quantum race.

We now proceed to show these two results.

\subsection{Approximate Nash equilibrium}

\begin{definition}
A two-player game described by payoff matrices $(A,B)$ is said to have an $\epsilon$-approximate Nash equilibrium 
$(p,q)$, for $\epsilon \ge 0$, if the following two conditions hold
\begin{align}
    p^{\T}Aq \geq v^{\T}Aq - \epsilon \mbox{ for all }  v \in \Delta_m \label{eq:ane1} \\
    p^{\T}Bq \geq p^{\T}Bu - \epsilon \mbox{ for all }  u \in \Delta_n \label{eq:ane2} \enspace .
\end{align}
\end{definition}

\begin{definition}
A two-player game described by payoff matrices $(A,B)$ is said to have an $\epsilon$-well supported Nash equilibrium 
$(p,q)$, for $\epsilon \ge 0$ if
\begin{align*}
    e_i^\T Aq \geq e_j^\T Aq - \epsilon \mbox{ for all }  i \in \sup(p) \mbox{ and } j \in [m] \\
    p^{\T}Be_i \geq p^{\T}Be_j - \epsilon \mbox{ for all }  i \in \sup(q) \mbox{ and } j \in [n] \enspace .
\end{align*}
\end{definition}
Note that an $\epsilon$-well supported Nash equilibrium is also an $\epsilon$-approximate Nash equilibrium, but the 
reverse does not hold.  

Before showing that the unique coinciding equilibrium in an $\ell$-dense symmetric stingy quantum race is a well supported 
equilibrium in the corresponding $\ell$-dense symmetric quantum race, we need the following claim.
\begin{claim}
\label{clm:easy_diff}
Let $p_1, \ldots, p_K$ be an $\ell$-dense sequence such that $K\ge 6\ell$.  Then for all $T^* \le i,j \le K$
\[
\left| \frac{p_i}{p_{i-1}} - \frac{p_j}{p_{j-1}} \right | \le \frac{14\ell}{K} \enspace .
\]
\end{claim}

\begin{proof}
\cref{p_T_lb_2} along with the $\ell$ dense condition gives us a lower bound on $p_{T^* -1} \ge \frac{5}{21} - \frac{1}{6} = \frac{1}{14}.$ Now note that for all $T^* \le i \le K,$ we have 
\[1\le \frac{p_i}{p_{i-1}} \le 1 + \frac{\ell}{Kp_{i-1}} \le 1 + \frac{\ell}{Kp_{T^*-1}} \le 1 + \frac{14\ell}{K},\] thus proving the claim. 
\end{proof}

\begin{theorem}
\label{thm:two_approx}
Let $p_1, \ldots, p_K$ be an $\ell$-dense sequence defining the symmetric stingy quantum race $(A_0, A_0^T)$ and the 
symmetric quantum race $(A,A^T)$.  Let $(x,x)$ be the unique coinciding Nash equilibrium for the stingy quantum race 
$(A_0, A_0^T)$ given by \cref{thm:symmetric}.  Then $(x,x)$ is a $\tfrac{7(\sqrt{2}-1)\ell}{K}$-well supported Nash equilibrium in 
the quantum race $(A,A^T)$. 
\end{theorem}

\begin{proof}
To show that $(x,x)$ is an $\epsilon$-well 
supported Nash equilibrium in the quantum race $(A,A^T)$ it suffices to show that 
$e_i^\T A x \ge e_j^\T A x - \epsilon$ for all $T^* \le i \le K$ and $j \in [K]$.

Note that 
 \begin{equation}
    e_i^{\T}A_0 x=
    \begin{cases}
      \frac{1}{z_{T^*}} & \text{if}\ T^* \leq i \leq K \\
      p_i & \text{if}\ i < T^*  \label{eq:ane3}
    \end{cases} \enspace .
  \end{equation}
As $A = A_0 + \tfrac{1}{2} \diag(p_1^2, \ldots, p_K^2)$ this means
 \begin{equation}
    e_i^{\T}Ax=
    \begin{cases}
      \frac{1}{z_{T^*}} + \frac{1}{2} p_i^2 x_i & \text{if}\ T^* \leq i \leq K \\
      p_i & \text{if}\ i < T^*  \label{eq:ane3}
    \end{cases} \enspace .
  \end{equation}
For $T^* \le i \le K$ and $j < T^*$ we have $e_i^\T A x \ge e_i^\T A_0 x \ge e_j^\T A_0 x = e_j^\T A x$.  

Now we show that for $T^* \le i, j \le K$ we have $e_i^\T A x \ge e_j^\T A x - \tfrac{7(\sqrt{2}-1)\ell}{z_{T^*}K}$.  
First we consider the case where $T^* < i \le K$ and $T^* \le j \le K$.  
In this case 
\begin{align*}
e_i^\T A x - e_j^\T A x &=  \frac{1}{2}\left( p_i^2 x_i - p_j^2 x_j \right)  \\
&\ge \frac{1}{2z_T^*}\left(\frac{p_i}{p_{i-1}} - \frac{p_j}{p_{j-1}}\right)\\
&\ge  -\frac{7(\sqrt{2}-1)\ell}{K} \enspace ,
\end{align*}
by \cref{clm:easy_diff} along with \cref{thm:tstar_upper}.  The second line holds with equality except for the case $j = T^*$,
where we have used the fact that $q_{T^*} \ge r_{T^*}$.  

Finally, consider the case where $i=T^*$ and $T^* < j \le K$.  Then
\begin{align*}
e_{T^*}^\T A x - e_j^\T A x &= \frac{1}{2z_{T^*}}\left(p_{T^*}^2 r_{T^*} - \frac{p_j}{p_{j-1}} +1 \right) \\
& \ge \frac{1}{2z_{T^*}}\left(1- \frac{p_j}{p_{j-1}}\right) \\
& \ge -\frac{7(\sqrt{2}-1)\ell}{K}. 
\end{align*}
\end{proof}

\subsection{Upper bound on payoff}
Let $(x,x)$ be the unique coinciding Nash equilibrium in an $\ell$-dense symmetric stingy quantum race $(A_0,A_0^\T)$.  We have just 
shown that $(x,x)$ is a $\frac{7(\sqrt{2}-1)\ell}{K}$-well supported Nash equilibrium in the corresponding quantum race $(A,A^\T)$.  By \cref{cor:payoff},
$(x,x)$ achieves payoff at least $\sqrt{2}-1 - 50\sqrt{2}(\sqrt{2}-1)^2 \frac{\ell}{K}$ in the game $(A,A^\T)$.  
In this section, we show 
that this payoff is within $O(\sqrt{\ell/K})$ of optimal among all symmetric strategies 
$(y,y)$ for the game $(A,A^\T)$.  

Our starting point is to use the program in \cref{eq:quad_prog_gen} to provide a means to upper bound the value 
of any symmetric equilibrium.

\begin{lemma} 
\label{lem:quad}
Let $(A,A^T)$ be a symmetric game and define for $c \ge 0$
\begin{equation*}
\begin{aligned}
\gamma_A(c) =\  & \underset{x}{\text{maximize}}
& & \frac{1}{2} x^T (A+A^T) x  \\
& \text{subject to}
& & A x \le c \mathbf{1},\\
&&& \mathbf{1}^{\T}x = 1, x \ge \mathbf{0}.
\end{aligned}
\end{equation*}
For all $c_0$, such that $\gamma_A(c) < c$ for all $c \ge c_0$, the payoff of any symmetric Nash equilibrium in 
the game $(A, A^\T)$ is less than $c_0$.
\end{lemma}

\begin{proof}
We show the contrapositive.  Suppose there is a symmetric Nash equilibrium $(x,x)$ with payoff $c \ge c_0$.  Then 
$(x,x,c,c)$ 
is a feasible solution to the program in \cref{eq:quad_prog_gen} with objective value $0$.  Thus 
$Ax \le c$ and $\frac{1}{2} x^\T (A+A^\T) x = c$.  
\end{proof}

This is the approach we take to upper bounding the payoff of symmetric Nash equilibria in a quantum race.

\begin{theorem}
\label{thm:payoff_upper}
Let $p_1, \ldots, p_K$ be an $\ell$-dense sequence with $K \ge 6\ell$.  Then any symmetric Nash equilibrium $(x,x)$ in the 
two-player quantum race defined by $p_1, \ldots, p_K$ has payoff at most $\sqrt{2}-1 + 5 \sqrt{\frac{\ell}{K}}$.
\end{theorem} 

\begin{proof}
Let $(A, A^\T)$ be the payoff matrices of a two-player quantum race defined by $p_1, \ldots, p_K$.  We will show that 
$\gamma_A(c) < c$ for all $c > \sqrt{2}-1 + 5 \sqrt{\frac{\ell}{K}}$. By \cref{lem:quad} this proves the theorem.

In the case of a quantum race $A+A^\T = p1^\T + 1p^\T - pp^\T$.  This means that over the probability simplex, the quadratic form 
$x^{\T}(A+A^{\T})x = - x^{\T}(pp^{\T}) x +2p^{\T}x$ is a negative-definite plus linear function.  In this case, $\gamma_A(c)$ is in the 
standard form of a quadratic program and has a dual program with matching value \cite{Dorn1960}.

\begin{equation}
\label{eq:dual}
\begin{aligned}
\gamma_A(c) = \ & \underset{v \in \mathbb{R}^K,\lambda,d \in \mathbb{R}}{\text{minimize}}
& & \frac{1}{2} \lambda^2 + c\cdot \mathbf{1}^\T v + d  \\
& \text{subject to}
& & A^T v \ge (1-\lambda) p - d\mathbf{1} \enspace , \\
&&& v \ge \mathbf{0}.
\end{aligned}
\end{equation}

Our approach will be to construct a feasible solution to \cref{eq:dual} to demonstrate that $\gamma_A(c) < c$ for all 
$c > \sqrt{2}-1 + 5\sqrt{\frac{\ell}{K}}$.

First note that for $c > \tfrac{1}{2}$ there is a trivial solution where $\lambda=1, d = 0$ and $v$ is the all-zero vector 
which shows that $\gamma(c) < c$.  We now focus on the case $c \le \tfrac{1}{2}$.  Let 
$\sqrt{2}-1 \le c \le \tfrac{1}{2}$.  We will develop a lower bound on $c$ which implies $\gamma_A(c) < c$.  

Let $S$ be the smallest index $i$ such that $p_i \ge c$.  Note that as $A$ is an $\ell$-dense quantum race we have 
$p_S \le c + \tfrac{\ell}{K}$.  We let 
\[
d = (1-\lambda)\left(1+p_K - \frac{p_K}{p_S} \right)
\]
and 
\[
v(i) = 
\begin{cases}
0 & \text{ if } 1 \le i < S \\
(1-\lambda - d) \frac{p_{S}}{\bar{p}_{S}}\frac{1}{p_i} \left(\frac{1}{p_i} - \frac{1}{p_{i+1}} \right) & \text{ if } S\le i < K \\
(1-\lambda -d) \frac{p_{S}}{p_K^2 \bar{p}_{S}} - \frac{(1-\lambda)}{p_K} & \text{ if } i = K \enspace .
\end{cases}
\]
The choice of $v$ comes from solving the system of linear equations $(A_0^\T v)_i = (1-\lambda)p_i -d$ for $S \le i \le K$.
The parameter $\lambda$ will be chosen later.

Let us see that $v$ satisfies the constraints of \cref{eq:dual}.  Note that 
\[
1-\lambda -d = (1-\lambda) p_K \frac{\bar{p}_S}{p_S}
\]
thus $v(K) =0$ and $v \ge 0$ so long as $\lambda \le 1$.  

As mentioned, by construction $v$ satisfies $(A_0^Tv)_i = (1-\lambda)p_i - d$ for $S \le i \le K$.  Thus as 
$A= A_0 + \tfrac{1}{2}\mathrm{diag}(p)^2$ and 
$v \ge 0$ we also have $(A^Tv)_i \ge (1-\lambda)p_i - d$ for $S \le i \le K$.

For $i < S$ we have that 
\begin{align*}
(A^T v)_i &\ge (A_0^T v)_i \\
&= \bar{p}_i p^T v \\
&\ge \bar{p}_{S} p^T v \\
&= (1-\lambda) p_{S} -d \\
&\ge (1-\lambda) p_i -d \enspace .
\end{align*}
Thus the constraint $A_0^T v \ge (1-\lambda)p -d\mathbf{1}$ is satisfied.

We have shown that $v$ is a feasible solution for any choice of $\lambda \le 1$.  We now choose $\lambda$ to minimize the 
objective value.  

Substituting our choices of $v,d$ into the objective value we have
\begin{align*}
\gamma_A(c) &\le \frac{1}{2}\lambda^2 + (1-\lambda)\left(1 + p_K - \frac{p_K}{p_S}\right) + c(1-\lambda)p_K 
\sum_{i=S}^{K-1} \frac{1}{p_i} \left(\frac{1}{p_i} - \frac{1}{p_{i+1}} \right) \\
&= \frac{1}{2}\lambda^2 + (1-\lambda) \left(1 + p_K\left(-\frac{\bar{p}_S}{p_S}+ 
c \cdot \sum_{i=S}^{K-1} \frac{1}{p_i} \left(\frac{1}{p_i} - \frac{1}{p_{i+1}} \right)  \right) \right)
\end{align*}

Define 
\[
\beta(c) = 1 + p_K\left(-\frac{\bar{p}_S}{p_S}+ 
c \cdot \sum_{i=S}^{K-1} \frac{1}{p_i} \left(\frac{1}{p_i} - \frac{1}{p_{i+1}} \right)  \right) \enspace .
\]
The objective value $\tfrac{1}{2}\lambda^2 + (1-\lambda) \beta(c)$ is minimized over $\lambda$ by taking 
$\lambda = \beta(c)$.  This makes the objective value $\beta(c) - \beta(c)^2/2$.  

We have now reduced the problem to showing
\[
\beta(c) - \beta(c)^2/2 - c < 0
\]
The roots of the corresponding quadratic equation are $1 \pm \sqrt{1-2c}$.  Note that $c \le 1/2$, thus the square root term will be 
real.  Thus we will simultaneously have $\beta(c) \le 1$ and $\beta(c) - \beta(c)^2/2 < c$ when $\beta(c) < 1 - \sqrt{1-2c}$.  
In \cref{clm:end_of_proof}, we show that $\beta(c) < 1 - \sqrt{1-2c^*}$ when 
$\sqrt{2}-1+ 5\sqrt{\frac{\ell}{K}} \le c \le \tfrac{1}{2}$.  
This will conclude the proof.  
\end{proof}

\begin{lemma}
\label{clm:end_of_proof}
$\beta(c) < 1 - \sqrt{1-2c}$ for any $\sqrt{2}-1 + 5\sqrt{\frac{\ell}{K}}  \le c \le \tfrac{1}{2}$.
\end{lemma} 

\begin{proof}
The claim is equivalent to showing
\begin{equation}
   \label{ineqcond}
\sum_{i=S}^{K-1} \frac{1}{p_i} \left(\frac{1}{p_i} - \frac{1}{p_{i+1}} \right) < \frac{1}{c} \left( \frac{\bar{p}_S}{p_S} - 
\frac{\sqrt{1-2c}}{p_K} \right) 
\enspace ,
\end{equation}
for all $\sqrt{2}-1 + 5 \sqrt{\frac{\ell}{K}} < c \le \frac{1}{2}$.

We proceed by a series of claims to simplify both the left hand side and the right hand side of 
\cref{ineqcond} until we 
can work the inequality into something easier to prove.  

We begin by simplifying the left hand side, for which we use the following claim.
\begin{claim}
\label{clm:lhs_helper}
Let $p_1, \ldots, p_K$ be an $\ell$-dense sequence.  If $K \ge 6 \ell$ then
\[
\sum_{i=T^*}^{K-1} \frac{1}{p_i} \left(\frac{1}{p_i} - \frac{1}{p_{i+1}} \right) \le \sqrt{2}+1 + 267\frac{\ell}{K} \enspace .
\]
\end{claim}

\begin{proof}
\begin{align*}
    \sum_{i=T^*}^{K-1} \frac{1}{p_i} \left(\frac{1}{p_i} - \frac{1}{p_{i+1}} \right) &= \sum_{i=T^*}^{K-1} \frac{1}{p_{i+1}} \left(\frac{1}{p_i} - \frac{1}{p_{i+1}} \right) + \sum_{i=T^*}^{K-1} \left(\frac{1}{p_i} - \frac{1}{p_{i+1}} \right)^2\\
    &\le z_{T^*} + \sigma(T^*) \\
    &\le \sqrt{2}+1 +  (50\sqrt{2}+196)\frac{\ell}{K} \enspace ,
\end{align*}
by \cref{thm:col} and \cref{cor:payoff}.
\end{proof}

We now use \cref{clm:lhs_helper} to upper bound the left hand side of \cref{ineqcond}.
\begin{claim}
\label{clm:lhs}
Let $\sqrt{2}-1 \le c \le p_S \le c + \tfrac{\ell}{K}$.  Then
\[
\sum_{i=S}^{K-1} \frac{1}{p_i} \left(\frac{1}{p_i} - \frac{1}{p_{i+1}} \right)  \le 
\frac{1}{c} \left(\frac{1}{c} - (1-c)(\sqrt{2}+1) \right) + 
287\frac{\ell}{K} \enspace .
\]
\end{claim}
\begin{proof}
\begin{align*}
\sum_{i=S}^{K-1} \frac{1}{p_i} \left(\frac{1}{p_i} - \frac{1}{p_{i+1}} \right)  
&=\sum_{i=T^*}^{K-1} \frac{1}{p_i} \left(\frac{1}{p_i} - \frac{1}{p_{i+1}} \right) -
\sum_{i=T^*}^{S-1} \frac{1}{p_i} \left(\frac{1}{p_i} - \frac{1}{p_{i+1}} \right) \\
&\le \sqrt{2}+1 + \frac{267\ell}{K} - \frac{1}{p_S} \sum_{i=T^*}^{S-1} \left(\frac{1}{p_i} - \frac{1}{p_{i+1}}\right) \\
&= \sqrt{2}+1 + \frac{267\ell}{K} - \frac{1}{p_S} \left(\frac{1}{p_{T^*}} - \frac{1}{p_S}\right) \\
&\le \sqrt{2}+1 + \frac{267\ell}{K} + \frac{1}{p_S^2} - \frac{\sqrt{2}+1}{p_S} + \frac{\ell(\sqrt{2}+1)^2}{p_S K} \\
&\le \frac{1}{p_S} \left(\frac{1}{p_S} - \bar p_S(\sqrt{2}+1) \right) + 
\frac{\ell}{K}\left(267+ \frac{(\sqrt{2}+1)^2}{p_S} \right) \\
&\le \frac{1}{c} \left(\frac{1}{c} - (1-c)(\sqrt{2}+1) \right) + 
\frac{\ell}{K}\left(267+ (\sqrt{2}+1)^2 + \frac{(\sqrt{2}+1)^2}{c} \right) \enspace .
\end{align*}
For the fourth line we used \cref{cor:ptstar} that $p_{T^*-1} \le \sqrt{2}-1$, together with $\ell$-density to get 
$p_{T^*} \le \sqrt{2}-1 + \tfrac{\ell}{K}$.  This then implies 
$\tfrac{1}{p_{T^*}} \ge \sqrt{2}+1 - \tfrac{\ell(\sqrt{2}+1)^2}{K}$ by \cref{eq:recip}.
\end{proof}

Now that we have obtained an upper bound on the left hand side of \cref{ineqcond}, we proceed to lower 
bound the right hand side.  We use the following claim to do this.

\begin{claim}
\label{clm:rhs}
Let $\sqrt{2}-1 \le c \le \tfrac{1}{2}$ and $c \le p_S \le c + \tfrac{\ell}{K}$.  Then
\[
\frac{1}{c} \left( \frac{\bar{p}_S}{p_S} - \frac{\sqrt{1-2c}}{p_K} \right) \ge 
\frac{1}{c} \left( \frac{1-c}{c} - \frac{\sqrt{1-2c}}{p_K} \right) - \frac{\ell}{c^3K} \left(1 -\frac{\ell}{K} \right) \enspace .
\]
\end{claim}

\begin{proof}
As $c \le p_S \le c + \tfrac{\ell}{K}$ we have
\[
\frac{\bar{p}_S}{p_S} \ge \frac{1-c-\frac{\ell}{K}}{c + \frac{\ell}{K}} \enspace .
\]
Applying \cref{eq:recip} to $\frac{1}{c + \ell/K}$ gives 
\begin{align*}
\frac{\bar{p}_S}{p_S} &\ge \left(1-c-\frac{\ell}{K} \right) \left(\frac{1}{c} - \frac{\ell}{Kc^2} \right) \\
&= \frac{1-c}{c} - \frac{(1-c)\ell}{Kc^2} - \frac{\ell}{Kc} + \frac{\ell^2}{K^2c^2} \\
&= \frac{1-c}{c} -\frac{\ell}{Kc^2} + \frac{\ell^2}{K^2c^2} \enspace .
\end{align*}
\end{proof}
Putting \cref{clm:lhs} and \cref{clm:rhs} together and simplifying, it now suffices to show
\begin{equation}
\label{eq:last_step}
0 < \sqrt{2}-c(\sqrt{2}+1) -\frac{\sqrt{1-2c}}{p_K} - \frac{\ell \tau}{K} \enspace, 
\end{equation}
for any $\sqrt{2}-1 + 5 \sqrt{\frac{\ell}{K}} < c \le \frac{1}{2}$, where $\tau = 150$.

We find a lower bound on the value of $c$ for which this inequality holds with the next claim.
\begin{claim}
\label{clm:last_ineq}
For any $\sqrt{2} -1 + \sqrt{\frac{2\ell(\sqrt{2}-1)^3 \left(\tau+\frac{6(\sqrt{2}-1)}{5}\right)}{K}} < c \le 1/2$ 
\begin{equation}
\label{eq:ineq_claim}
0 < \sqrt{2} - (\sqrt{2}+1)c - \frac{\sqrt{1-2c}}{p_K} - \frac{\ell \tau}{K} \enspace .
\end{equation}
\end{claim}

\begin{proof}
We will use \cref{eq:minus_sqrt} to bound $\sqrt{a-c}$, specifically in our context the inequality
\begin{equation}
\label{eq:sp_sqrt}
\sqrt{3-2\sqrt{2}- 2\delta} \le \sqrt{2}-1- (\sqrt{2}+1)\delta- \frac{1}{2}(\sqrt{2}+1)^3 \delta^2 \enspace .
\end{equation}
Consider some $\delta \ge 0.$ By substituting $c = \sqrt{2}-1 + \delta$ into the left hand side of \cref{eq:ineq_claim} and applying 
\cref{eq:sp_sqrt} we obtain
\begin{align*}
\sqrt{2} - (\sqrt{2}+1)c - \frac{\sqrt{1-2c}}{p_K} - \frac{\ell \tau}{K} &= \sqrt{2} -1 - (\sqrt{2}+1)\delta - \frac{\sqrt{3-2\sqrt{2}-2\delta}}{p_K} - \frac{\ell \tau}{K} \\
&\ge (\sqrt{2}-1 -(\sqrt{2}+1)\delta)\left(1-\frac{1}{p_K}\right) +\frac{\delta^2(\sqrt{2}+1)^3}{2p_K} - \frac{\ell \tau}{K} \\
&\ge -\left(\sqrt{2}-1 -(\sqrt{2}+1)\delta\right)\frac{6\ell}{5K} +\frac{\delta^2(\sqrt{2}+1)^3}{2} - \frac{\ell \tau}{K} \\
&\ge -\frac{6(\sqrt{2}-1)\ell}{5K} +\frac{\delta^2(\sqrt{2}+1)^3}{2} - \frac{\ell \tau}{K}. \enspace 
\end{align*}
For the second to last inequality, we use the observation that $\frac{1}{p_K} \leq 1 + \frac{6\ell}{5K}$, which follows from the condition $K \ge 6\ell$.  The last expression will be positive for 
\[
\delta > \sqrt{\frac{2\ell(\sqrt{2}-1)^3 \left(\tau+\frac{6(\sqrt{2}-1)}{5}\right)}{K}} \enspace,
\]
giving the claim. 
\end{proof}
Applying \cref{clm:last_ineq} with the value $\tau = 150$, we see that 
\cref{eq:last_step} will hold for any 
\[
\sqrt{2}-1 + 5 \sqrt{\frac{\ell}{K}} < c \le \frac{1}{2} \enspace .
\]
This finishes the proof of \cref{clm:end_of_proof}.
\end{proof}

Let us summarize the results of this section.  Let $(A_0,A_0^T)$ be a $\ell$-dense symmetric stingy quantum race and let $(A,A^T)$ be the 
corresponding quantum race, and suppose that $\frac{\ell}{K} \le \frac{1}{6}$.  \cref{thm:payoff_upper} shows that the payoff of any 
symmetric strategy $(y,y)$ in the game $(A,A^T)$ is at most $\sqrt{2}-1+5\sqrt{\ell/{K}}$.  On the other hand, letting $(x,x)$ be the unique 
coinciding equilibrium in the game $(A_0,A_0)$, this strategy achieves payoff at least $\sqrt{2}-1 - O(\ell/K)$ in the game $(A,A^T)$ by 
\cref{cor:payoff}.  Thus $(x,x)$ is an approximate Nash equilibrium for the game $(A,A^T)$ and moreover is nearly optimal 
in terms of payoff (within $O(\sqrt{\ell/K})$ amongst all symmetric strategies.

\section{Multiplayer quantum races}
\label{sec:multi}
\subsection{Basic properties}
For an integer $n \geq 2,$ an {\em $n$-player game} 
is specified by a set of {\em pure strategies} $S_i$, and 
{\em payoff} functions $u_i : S \rightarrow
\mathbb{R},$ for each player $i \in [n]$, where by definition $S = S_1 \times
\cdots \times S_n$ is the set of {\em pure strategy profiles}. For
$s \in S$, the value $u_i(s)$ is the payoff of player $i$ for pure
strategy profile $s$. Let $S_{-i} =  S_1 \times \cdots \times
S_{i-1} \times S_{i+1} \times \cdots \times S_n$ be the set of all
pure strategy profiles of players other than $i$. For $s \in S$ and $i \in [n]$, we
set the {\em partial} pure strategy profile $s_{-i} \in S_{-i}$ to be $(s_1 ,
\ldots s_{i-1}, s_{i+1}, \ldots , s_n)$. For $s'$ in $S_{-i}$,
and $s_i \in S_i$, we denote by $(s', s_i)$ the {\em combined}
pure strategy profile $(s'_1, \ldots , s'_{i-1}, s_i, s'_{i+1},
\ldots, s'_n) \in S$.
We will suppose that each player has
$m$ pure strategies 
and that $S_i =  \{ e_1, \ldots , e_m \}$, the canonical basis of the
vector space $\mathbb{R}^m$, for all $i \in [n]$, and
therefore $S = \{ e_1, \ldots , e_m \}^n$. For simplicity, instead of $e_j$ we often say strategy $j$.

A {\em mixed strategy} for player $i$ is a probability distribution
over $S_i$ that we identify with a vector $x_i = (x_i^1, \ldots x_i^m)$ such that
$x_i^j \geq 0,$ for all $j \in [m]$, and $\sum_{j \in [m]} x_i^j = 1$. 
We denote by $\Delta_i$ the set of mixed strategies for
$i$, and we call $\Delta = \Delta_1 \times \cdots \times \Delta_r$
the set of {\em mixed strategy profiles}. 
For a mixed strategy profile $x = (x_1, \ldots , x_n)$ and pure strategy profile $s \in
S$, the product $x^s = x_1^{s_1} x_2^{s_2} \cdots x_n^{s_n}$ 
is the probability of $s$ in $x$. We will consider the multilinear
extension of the payoff functions from $S$ to $\Delta$ defined by
$u_i(x) = \sum_{s \in S} x^s u_i(s)$. 
The set $\Delta_{-i}$, the
partial mixed strategy profile $x_{-i}$, for $x \in \Delta$ and $i \in [n]$,  and the
combined mixed strategy profile $(x', x_i)$ for $x' \in \Delta_{-i}$
and $x_i \in \Delta_i$ are defined analogously to the pure case. 

The pure strategy $s_i$ is a {\em best response} for player $i$ against
the partial mixed strategy profile $x' \in \Delta_{-i}$ if it maximizes
$u_i(x', \cdot).$ For $x \in \Delta$ and $i \in [n]$, we will denote by $\br ( x_{-i} )$ the set
of best responses of player $i$ against $x_{-i}$.
A {\em Nash equilibrium} is a mixed strategy profile $x=(x_1, \ldots , x_n)$ such that 
$\sup(x_i) \subseteq \br ( x_{-i} )$, for all $i \in [n]$.

\begin{definition}[$n$-party stingy quantum race]
Let $n\geq2$ be a positive integer. 
An $n$-{\em party stingy quantum race} is defined by a sequence of increasing probabilities
$0  < P_1 < P_2 < \ldots < P_K \leq 1$, for some positive integer $K$. The set of pure strategies
of all players is $[K]$. For every $i$, the utility function of the $i^{th}$ player is defined as
$$
u_i(s_1, \ldots , s_n) = P_{s_i} \prod_{k\neq i, s_k \leq s_i}\bar{P}_{s_k}.
$$
\end{definition}

Consider an $n$-party stingy quantum race  given by the probabilities 
$0 < P_1  < \ldots < P_K \leq 1$,
and let $x_{-i} \in \Delta_{-i}$ for some $i \in [n]$. If player $i$ 
plays the pure strategy $s$ against $x_{-i}$,
her payoff is 
\begin{equation}
\label{eq:develop}
u_i (x_{-i},s) = P_{s} \prod_{k\neq i}
\left(\sum_{s_k \in \sup_{\leq s}(x_k)} x_k^{s_k} \bar{P}_{s_k} + 
\sum_{s_k \in \sup_{> s}(x_k)} x_k^{s_k}  \right).
\end{equation}

The following is the multiparty analog of \cref{cla:interval}.
\begin{claim}
\label{cla:multi-interval}
Let $x= (x_1, \ldots x_n)$ be a Nash equilibrium of an $n$-party stingy quantum race 
defined by probabilities $P_1< \ldots < P_K$.
If  $s_1 \in \sup (x_i)$, for some $i \in [n]$, then for all 
$s_2 > s_1$ there exists $k \ne i$ such that $\sup_{\leq s_1}(x_k) \ne \sup_{\leq s_2}(x_k)$.
\end{claim}

\begin{proof}
If $\sup_{\leq s_1}(x_k) = \sup_{\leq s_2}(x_k)$, for all $k\neq i$ then
$$
u_i(x_{-i}, s_2) = \frac{P_{s_2}}{P_{s_1}} u_i(x_{-i}, s_1).
$$
As $p_{s_1} < p_{s_2}$,  the payoff for playing 
$s_2$ is strictly larger than that for playing $s_1$.  
Therefore $s_1$ is not a best response for $x_{-i}$, in contradiction with the definition of a Nash 
equilibrium.
\end{proof}

This claim implies the following properties for the supports of Nash equilibria in a multiplayer stingy 
quantum race.
\begin{corollary}
Let $x$ be a Nash equilibrium of an $n$-party stingy quantum race defined by probabilities 
$0  < P_1 < P_2 < \ldots < P_K \leq 1$. Then we have:
\begin{itemize}
   \item $\bigcup_{i=1}^{n} \sup(x_i)$  is an interval containing $K$,
   \item for every $i \in [n],$ for every $s_1, s_2 \in \sup(x_i)$ there exists 
   $s \in \bigcup_{k \neq i} \sup(x_k)$ with $s_1 < s \le s_2$.  
\end{itemize}
\label{cor:multi-interval}
\end{corollary}

Let $x$ be a Nash equilibrium of an $n$-party stingy quantum race. We say that $x$ is 
{\em coinciding} if $\sup(x_i) = \sup(x_k)$, for all $i,k \in [n].$ 
We call this common support in a coinciding Nash equilibrium the {\em support} of the equilibrium.  
In the multiparty case we will only study coinciding Nash equilibria.

\subsection{Coinciding Nash equilibria of stingy multiplayer races}
By \cref{cor:multi-interval} we know that in a coinciding Nash equilibrium of an
$n$-party stingy quantum race
the support of the equilibrium is of the form $\{T, T+1, \ldots, K\}$, for some 
$1 \leq T \leq K$.  
We would like to characterize these coinciding equilibria.

\begin{lemma}
\label{lem:mcharacterize}
Let $x =(x_1, \ldots , x_n)$, where $x_i$ is a $K$-dimensional real vector for every 
$i \in [n ]$, and let $1 \leq T \leq K$. Then
$x$ is a Nash equilibrium  of support $\{T, T+1, \ldots, K\}$ 
in an $n$-party stingy quantum race defined by 
$0  < P_1 < \ldots < P_K \leq 1$
if and only $x$ satisfies the following system, for all $i \in [n]$:
\begin{align}
u_i(x_{-i}, t) & = u_i(x_{-i}, T)~  & \text{for} ~~ T < t \leq K ~, \label{eq:mne1} \\
u_i(x_{-i}, T-1) & \leq u_i(x_{-i}, T) ~,  & \label{eq:mne2} \\
x_i^t   & = 0  & \text{for} ~~ 0 < t < T ~, \label{eq:mzero} \\
x_i^t   & > 0  & \text{for} ~~ T \leq t \leq K ~, \label{eq:mne3} \\
\sum_{t = T}^K x_i^t & = 1 ~ . & \label{eq:mne4}
\end{align}

\end{lemma}
\begin{proof}
For every $i \in [n]$, \cref{eq:mzero}--(\ref{eq:mne4}) express that $x_i$ is a probability 
distribution of support $\{T, T+1, \ldots, K\}$.
For $T\geq 2$, when playing  a strategy $t < T$ against the partial mixed strategy profile $x_{-i}$, the
$i$th player's payoff 
is maximized if she plays \ $T-1$. 
Therefore \cref{eq:mne1} and~(\ref{eq:mne2}) express that the strategies in her support are all
best responses against $x_{-i}$. 
\end{proof}

\begin{definition}
For an $n$-party stingy quantum race defined by 
$0  < P_1 < \ldots < P_K \leq 1$ we define its {\em reduced game} as the 2-party
stingy quantum race defined by the two sequences of probabilities $p_1  < \ldots < p_K$, 
and  $P_1 < \ldots < P_K$  where
$p_j = P_j^{1/(n-1)}$, for $1 \leq j \leq K$. 
\end{definition}
We denote by $A$  the payoff matrix of the first player in the reduced game

\begin{lemma}
\label{lem:reduction}
Let an $n$-party stingy quantum race be defined by 
$0  < P_1 < \ldots < P_K \leq 1$, let
$x = (x_1, \ldots , x_n)$, where $x_i$ is a $K$-dimensional vector, and let $1 \leq T \leq K$.
Then \cref{eq:mne1}--(\ref{eq:mne4}) are satisfied by $x$, for every $i \in [n]$ 
if and only if  \cref{eq:ne1}--(\ref{eq:ne4}) for the reduced game
are satisfied by $x_i$, for every $i \in [n]$.
\end{lemma}

\begin{proof}
First we do the ``if'' direction.  By \cref{lem:unicity}, there is a unique solution to \cref{eq:ne1},~(\ref{eq:zero}) and~(\ref{eq:ne4}),
thus if these are satisfied by each $x_i$ then we must have $x_1 = \cdots = x_n$.  Let $y$ be this common vector. Then for every 
$t \in [T-1, K]$,  
\begin{align*}
u_i (x_{-i},t) &=  P_t \prod_{k \ne i} \left(\sum_{s_k \in \sup_{\leq t}(x_k)} x_k^{s_k} \bar P_{s_k} + 
\sum_{s_k \in \sup_{>t}(x_k)} x_k^{s_k} \right) \\
&= P_t \left(\frac{e_t^\T A y}{p_t} \right)^{n-1} \\
&= \left(e_t^\T A y\right)^{n-1} \enspace .
\end{align*}

We now prove the ``only if'' direction.
Let us suppose that $x$ satisfies \cref{eq:mne1}--(\ref{eq:mne4}).
\cref{eq:zero}--(\ref{eq:ne4}) are respectively identical to 
\cref{eq:mzero}--(\ref{eq:mne4}),
and therefore are satisfied. Let us consider \cref{eq:ne1}.  We fix some $T < t \leq K$.
From \cref{eq:develop} we have, for every $i \in [n]$,
\begin{equation}
\label{eq:little}
u_i (x_{-i},t) = P_{t} \prod_{k\neq i}
\left(\sum_{j \leq t} x_k^{j} \bar{P}_{j} + 
\sum_{j > t}   x_k^{j}  \right), 
\end{equation}
and
\begin{equation}
\label{eq:big}
u_i (x_{-i},T) = P_{T} \prod_{k\neq i}
\left( x_k^{T} \bar{P}_{T} + 
\sum_{j > T}   x_k^{j}  \right).
\end{equation}
Therefore
\begin{align*}
 \left(  e_t^{\T}Ax_i \right)^{n-1} &= 
 p_{t}^{n-1} \left(\sum_{j \leq t} x_i^{j} \bar{P}_{j} +  \sum_{j > t}   x_i^{j}  \right)^{n-1} \\
& = \frac{p_t^{n-1}P_{t}^{n-1}}{P_t^{n}} \frac{1}{u_i(x_{-i}, t)^{n-1}} \prod_{\ell=1}^n  u_{\ell}(x_{-\ell}, t) \\
& = \frac{1}{u_i(x_{-i}, t)^{n-1}} \prod_{\ell=1}^n  u_{\ell}(x_{-\ell}, t)  \\
&= \frac{1}{u_i(x_{-i}, T)^{n-1}} \prod_{\ell=1}^n  u_{\ell}(x_{-\ell}, T) \\
& =  \left(  e_T^{\T}Ax_i \right)^{n-1}  \enspace .
\end{align*}
The second equality follows from the fact that 
$$\prod_{\ell=1}^n  u_{\ell}(x_{-\ell}, t) = P_t^n \prod_{k= 1}^n  \left(\sum_{j \leq t} x_k^{j} \bar{P}_{j} + 
\sum_{j > t}   x_k^{j}  \right)^{n-1} ,$$
for the third equality we used $p_{t-1}^{n-1}=P_{t-1}$, and for the fourth equality we used \cref{eq:mne1}.  
As $e_t^{\T}Ax_i, e_T^{\T}Ax_i$ are real and nonnegative, it follows that $e_t^{\T} A x_i = e_T^{\T}A x_i$ for all $T < t \le K$, 
establishing that \cref{eq:ne1} holds. 

Now we turn to \cref{eq:ne2}. Since for every $i \in [n]$, 
\cref{eq:ne1},~(\ref{eq:zero}) and~(\ref{eq:ne4}) are
satisfied by $x_i$, from \cref{lem:unicity} we know that $x_1 = \dots = x_n$. Thus for every $i \in [n]$,
we have
\begin{align*}
 \left(  e_{T-1}^{\T}Ax_i \right)^{n-1} &= 
p_{T-1}^{n-1} \prod_{k\neq i} \left(   \sum_{j = T}^K   x_k^{j}  \right) \\
&= u_i (x_{-i},T-1)  \\
&  \leq u_i(x_{-i}, T) \\
& = P_{T} \prod_{k\neq i} \left( x_k^{T} \bar{P}_{T} + \sum_{j > T}   x_k^{j}  \right) \\
& = \left(  e_T^{\T}Ax_i \right)^{n-1}  \enspace ,
\end{align*}
where for the first equality we used $x_1 = \dots = x_n$, for the second equality we used $p_{t-1}^{n-1}=P_{T-1}$, and
for the inequality we used \cref{eq:mne2}.
The statement of the Lemma therefore follows.
\end{proof}

\begin{theorem}
\label{thm:msymmetric}
An $n$-party stingy quantum race always has a unique coinciding Nash equilibrium
$x=(x_1, \ldots, x_n)$, where $x_1 = \cdots = x_n$.
If the game is defined by the probabilities 
$0 < P_1 < P_2 < \cdots < P_K$ then
the coinciding equilibrium 
has support $\{T^*, T^*+1,  \ldots, K\}$, where $T^*=T_B^*$ of the reduced game, 
and for all $i \in [n]$, the distribution $x_i$  is defined on its support as
$$
x_i^t =
\begin{cases}
{r_{T^*}^B}/{z_{T^*}^B}& \text{if} ~~~ t= T^*,\\
{q_t^B}/{z_{T^*}^B}& \text{if} ~~~ T^* < t \leq K. 
\end{cases}
$$
\end{theorem}

\begin{proof}
Combining \cref{lem:mcharacterize} and \cref{lem:reduction}, we get that 
$x$ is a coinciding Nash equilibrium of support $\{T, T+1,  \ldots, K\}$
if and only if $x_i $ satisfies \cref{eq:ne1}--(\ref{eq:ne4})
for the reduced game, for all $i \in [n]$.
By \cref{thm:unique1-5} this happens if and only if 
$T = T^*_B$ of the reduced game, and the unique solution for $x_i$, for $i \in [n]$, is the one stated by the Theorem.
\end{proof}


\subsection{Collision probability of the stingy coinciding equilibrium}
Our main objective in this section is to upper bound the collision probability---the probability that 
two or more players succeed at the same time---in the coinciding equilibrium 
found in the last section for an $\ell$-dense stingy $n$-player quantum race.  To help with this, we make 
the following definition.
\begin{definition}
For a joint probability distribution $y = (y_1, \ldots, y_n) \in \Delta$, let 
$\cp_i^m(y)$ denote the probability that player $i$ succeeds first and that exactly $m$ players 
(including $i$) succeed at the same time under the joint strategy $y$.  Let $\cp_i(y)=\sum_{m=2}^n \cp_i^m(y)$ 
denote the probability that player $i$ succeeds first and at least one other player succeeds at the same time.
\end{definition}
 
Let us also set up notation to describe the coinciding equilibrium in a stingy multiplayer race.  Define
\begin{align*}
q_i &= \frac{1}{P_i} \left(\frac{1}{P_{i-1}^{1/(n-1)}} - \frac{1}{P_i^{1/(n-1)}} \right) \\
r_T &= \frac{1}{\bar{P}_T}\left(\frac{1}{P_K^{1/(n-1)}} - \sum_{i=T+1}^K \bar P_i q_i \right) \\
z_T &= r_T + \sum_{i=T+1}^K q_i \enspace .
\end{align*}
Let $T^*$ be the starting point of the support of the coinciding equilibrium.  Then by 
\cref{thm:msymmetric}, the strategy of player $i$ in the coinciding equilibrium is given by
\begin{equation}
\label{eq:def_multiparty}
x_i^t =
\begin{cases}
{r_{T^*}}/{z_{T^*}}& \text{if} ~~~ t= T^*,\\
{q_t}/{z_{T^*}}& \text{if} ~~~ T^* < t \leq K, \\
0 & \text{if} ~~~ t < T^* \enspace .
\end{cases}
\end{equation}

To obtain concrete bounds on the collision probability, we will need bounds on $z_{T^*}$ and $P_{T^*-1}$.
\begin{lemma}
\label{payoffmul}
In any stingy multiplayer quantum race with $n$ players $(1/z_{T^*})^{n-1} < 1/n$.  
\end{lemma} 

\begin{proof}
The probability that any given player wins the race is $(1/z_T^*)^{n-1}$.  As in a stingy race no two 
players can win at the same time, the probability that the race has a winner is exactly 
$n/z_{T^*}^{n-1}$.  The probability that no one wins the race, which occurs if two or more players 
succeed at the same time, or if no one succeeds, is positive, thus $(1/z_{T^*})^{n-1} < 1/n$.
\end{proof}

\begin{restatable}{theorem}{boundPTstar}
\label{lem:bound_PTstar}
Let $P_1, \ldots, P_K$ define a stingy $n$-player quantum race with $n \ge 2$.  Then $P_{T^*-1} < \frac{1}{n}$.  
If in addition $P_1,....,P_K$ form an $\ell$-dense sequence and $4\e \ell n \le K$ then 
$P_{T^*-1} \geq \tfrac{1}{2\e n}$, where $\e$ is Euler's number.
\end{restatable}

\begin{proof}
As $T^*-1$ is not in the support of the coinciding Nash equilibrium, it must be the case that $P_{T^*-1} \le \left(\frac{1}{z_{T^*}}\right)^{n-1} < \frac{1}{n}$, by \cref{payoffmul}.

Now we turn to the lower bound.
If $P_{T^*+1} > \tfrac{1}{n}$ then as $P_1, \ldots, P_K$ is $\ell$-dense we must have 
$P_{T^*-1} \ge \tfrac{1}{n} - \tfrac{2\ell}{K} \ge \tfrac{1}{2\e n}$ and the proof is done.  
Otherwise, define $S = \max_{i > T^*} \{i : P_i \le \tfrac{1}{n}\}$. Note that by definition and 
$\ell$-density we have $\frac{1}{n} \geq p_S \geq \frac{1}{n} - \frac{\ell}{K}$. Now we know from the positivity condition of $r_{T*}$ that    

\begin{align*}
    1 & > \sum_{T^*+1}^K \bar{P}_iq_i \\
    & \geq \sum_{T^*+1}^S \frac{1-P_i}{P_i} \left( P_{i-1}^{-\frac{1}{n-1}}- P_{i}^{-\frac{1}{n-1}}\right) \\
    & \geq  \frac{1-P_S}{P_S}\left( P_{T^*}^{-\frac{1}{n-1}} - P_{S}^{-\frac{1}{n-1}} \right) \\
    & \geq (n-1) \left(P_{T^*}^{-\frac{1}{n-1}} - \left(\frac{Kn}{K-\ell n}\right)^{\frac{1}{n-1}}\right) \\
    & \geq (n-1) \left(P_{T^*}^{-\frac{1}{n-1}} - \left(\frac{4\e n}{4\e-1}\right)^{\frac{1}{n-1}}\right) \enspace .
\end{align*}
For the last inequality we have used that $\ell n \le \frac{K}{4\e}$.  Rearranging the last inequality we have
\[
P_{T*} \geq \left(\frac{n-1}{1+ (n-1)\left(1+\frac{4\e n}{4\e-1}\right)^{\frac{1}{n-1}}}\right)^{n-1} \enspace .
\]
This means
\begin{align*}
\ln(P_T^*) &\ge (n-1) \left(\ln(n-1) - \ln\left(1+ (n-1)\left(\frac{4\e n}{4\e-1}\right)^{1/(n-1)}\right) \right) \\
&=-\ln\left(\frac{4\e n}{4\e -1}\right) - (n-1)\ln\left(1+ (n-1)^{-1}\left(\frac{4\e n}{4\e -1}\right)^{-1/(n-1)}\right) \enspace.
\end{align*}
It now remains to upper bound
\[
(n-1)\ln\left(1+ (n-1)^{-1}\left(\frac{4\e n}{4\e -1}\right)^{-1/(n-1)}\right) \le \left(\frac{4\e n}{4\e-1}\right)^{-1/(n-1)} \le 1 \enspace,
\]
for $n \ge 2$, where we have used $\ln(1+x) \le x$ to obtain the first inequality.
Putting it together we have $\ln(P_{T^*}) \ge -\ln(\frac{4en}{4\e-1}) -1$, which implies 
$P_{T^*-1} \ge P_{T^*} - \frac{\ell}{K} \ge \frac{1}{2\e n}$.
\end{proof}

The next lemma bounds the collision probability for player $i$ when all players but player $i$ play according to the 
coinciding equilibrium, and player $i$ plays an arbitrary strategy $v$.  We will use this lemma to bound the 
total collision probability and also in \cref{sec:multiplayer_races} to show that the stingy coinciding equilibrium
is an approximate equilibrium in a multiparty race.
\begin{lemma}
\label{lemcp}
Let $P_1,....,P_K$ define an $\ell$-dense stingy $n$-player quantum race, $n \ge 2$, with 
$4\e n\ell \leq K$, and let 
$x$ be the unique coinciding equilibrium given by \cref{eq:def_multiparty}.  Then 
$\mathrm{cp}_i(x_{-i},v) \leq \frac{8\e \ell}{K}$ for any $i \in [n]$ and $v \in \Delta_i$. 
\end{lemma}

\begin{proof}
We use the notation from the definition of the coinciding equilibrium in \cref{eq:def_multiparty}.  Note that since 
$p_{T^*} \ge \frac{1}{2\e n}$ by \cref{lemcp}, and we assume that $\frac{\ell}{K} \le \frac{1}{4\e n}$, it must be the case that 
$T^* \ge 2$, and therefore it makes sense to talk about $T^*-1$.  Further, note that $z_{T^*} \ge 1$ by \cref{payoffmul}.

We upper bound $\cp_i(x_{-i},v)$ by multiplying the probability that player $i$ and player $j$ succeed at the same time (not 
necessarily first) by $n-1$.  
\begin{align*}
\cp_i(x_{-i},v) &\leq  \frac{n-1}{z_T^*}\left( P_{T^*}^2 v_{T^*} r_{T^*} +\sum_{t = T^*+1}^K P_t^2 v_tq_t \right)\\
&\le \frac{n-1}{z_T^*} \sum_{t = T^*}^K P_t^2 v_tq_t  \\
&= \frac{n-1}{z_T^*} \sum_{t = T^*}^K P_t v_t \left(\frac{1}{P_{t-1}^{\frac{1}{n-1}}} - \frac{1}{P_{t}^{\frac{1}{n-1}}} \right)\\
&= \frac{n-1}{z_T^*} \sum_{t = T^*}^K P_t v_t 
\frac{P_t^{\frac{1}{n-1}}-P_{t-1}^{\frac{1}{n-1}}}{(P_{t-1}P_t)^{\frac{1}{n-1}}} \\
\end{align*}
For notational convenience, let $f(x) = x^{\frac{1}{n-1}}$.  By the concavity of $f$ on $[0,\infty)$ we can 
upper bound $f(P_t) - f(P_{t-1})$ by
\begin{equation}
\label{pdiff}
f(P_i) - f(P_{i-1}) \leq (P_i-P_{i-1}) f'(P_{i-1}) \leq \frac{\ell P_{i-1}^{\frac{2-n}{n-1}}}{K(n-1)} \enspace .
\end{equation} 
Using this upper bound, we can continue:
\begin{align*}
\mathrm{cp}_i(x_{-i},v)  &\le \frac{\ell}{Kz_{T^*}} 
\sum_{t=T^*}^K P_t v_t \frac{P_{t-1}^{\frac{2-n}{n-1}}}{(P_{t-1}P_t)^{\frac{1}{n-1}}} \\
&= \frac{\ell}{Kz_{T^*}} \sum_{t=T^*}^K v_t \frac{1}{P_{t-1}^{\frac{1}{n-1}}}\left(\frac{P_t}{P_{t-1}}\right)^{\frac{n-2}{n-1}}\\
&\le \frac{4\e \ell}{Kz_{T^*}} \sum_{t=T^*}^K v_t \left(\frac{P_t}{P_{t-1}}\right)^{\frac{n-2}{n-1}} \enspace, \\
\end{align*}
as $(2\e n)^{1/(n-1)} \le 4 \e$ for all $n \ge 2$.  Continuing, we have
\begin{align*}
\mathrm{cp}_i(x_{-i},v)&\le \frac{4\e \ell}{Kz_{T^*}} \sum_{t=T^*}^K v_t \left(1 + \frac{\ell}{KP_{t-1}}\right)^{\frac{n-2}{n-1}} \\
&\le \frac{4\e \ell}{Kz_{T^*}} \left(1 + \frac{2\e n\ell}{K}\right)^{\frac{n-2}{n-1}} \\
&\le \frac{8\e \ell}{K} \enspace .
\end{align*}
\end{proof}

\colprobm*

\begin{proof}
The probability that two or more players succeed at the same time is at most
\[
\sum_{i=1}^n \cp_i(x) \le \frac{8\e n\ell}{K} \enspace,
\]
by \cref{lemcp}.
\end{proof}
 
\subsection{Multiplayer quantum races}
\label{sec:multiplayer_races}
In this section we use our results about the stingy multiplayer quantum race to analyse the multiplayer quantum race.  
Namely, we show that the coinciding equilibrium in a stingy multiplayer quantum race is an approximate equilibrium 
in a multiplayer race.  The difference between a stingy multiplayer race and a multiplayer race is that in a 
multiplayer race, the payoff is equally divided amongst all players who succeed first at the same time.  
\begin{definition}[Multiplayer quantum race]
\label{def:multiplayer_race}
Let $u_i$ be the payoff function for player $i \in [n]$ in the $n$-player stingy race defined by $P_1 < \cdots < P_k$, 
as given by \cref{eq:develop}.  The payoff function $u_i'$ in the $n$-player quantum race defined by 
$P_1, \ldots, P_k$ is 
\[
u_i'(x) = u_i(x) + \sum_{m=2}^n \frac{\cp_i^m(x)}{m} \enspace .
\]
\end{definition}
While the tie-splitting payoff in \cref{def:multiplayer_race} is quite natural, one could imagine other 
definitions in-between stingy multiplayer races and the definition of multiplayer races we have given.  Our results in this 
section depend very weakly on the exact definition of how ties are split in a multiplayer race.  In fact, the only
property we use is 
\begin{equation} 
\label{eq:amne1}
u_i'(x)  \leq u_i(x) + \mathrm{cp}_i(x) \enspace .
\end{equation}
This property holds under any reasonable definition of tie-splitting.

Now we show \cref{thm:mapprox} from the introduction that the coinciding Nash equilibrium in a multiplayer stingy 
quantum race is an approximate Nash equilibrium in a multiplayer quantum race.

\mapprox*

\begin{proof}
Since $x$ is a Nash equilibrium in the stingy race, for any $i \in [n]$ and probability vector $v \in \Delta_i$ we have
\begin{equation} 
\label{eq:amne2}
u_i(x) \geq u_i(x_{-i},v) \enspace .
\end{equation}
By \cref{def:multiplayer_race} and \cref{eq:amne1} we have 
\begin{align*}
u_i'(x) \ge u_i(x) &\ge u_i'(x_{-i},v) - \mathrm{cp}_i(x_{-i},v) \\
&\ge u_i'(x_{-i},v) - \frac{8\e \ell}{K} \enspace,
\end{align*}
by \cref{lemcp}.
\end{proof}

\section{Conclusion}
We conclude with some open problems.  
\begin{itemize}
\item Can we characterize all Nash equilibria in the multiparty stingy quantum race?  We have given a coinciding Nash 
equilibrium and shown that it is unique, however, it would be interesting to classify the Nash equilibria where players play 
strategies with different supports.
\item Can we show that the approximate Nash equilibrium we have given for a multiparty quantum race is nearly optimal 
in terms of payoff among all symmetric Nash equilibria?  We were able to do this in the two-party case, however, this proof 
used tools like quadratic programming that are not available in the multiparty case.
\item Our model of quantum races assumes that if a player measures and is unsuccessful then she simply stops playing.  A more realistic 
scenario is that, if a player measures and fails, she immediately restarts and tries again.  Do the Nash equilibria in this version of 
the game change substantially?  The large number of available strategies in this version of the game make it much more difficult 
to analyze.
\end{itemize}

\section*{Acknowledgments}
We would like to thank Gavin Brennen, Hartmut Klauck, and Marco Tomamichel for very useful 
discussions about game theory and quantum Bitcoin mining.  We would also like to thank Or Sattath 
for sharing an early version of his work \cite{Sat18} and discussions about its implications.

Part of this work was done while Troy Lee was at the School of Physical and Mathematical Sciences, Nanyang 
Technological University and the Centre for Quantum Technologies, Singapore.  Troy Lee and Maharshi Ray 
were supported in part by the Singapore National Research 
Fellowship under  NRF RF Award No.\  NRF-NRFF2013-13.  This research was further
partially funded by the Singapore Ministry of Education and the National Research Foundation
under grant R-710-000-012-135, and in part by the QuantERA ERA-NET Cofund project QuantAlgo.


\newcommand{\etalchar}[1]{$^{#1}$}

\appendix

\section{Symmetric 2-player stingy quantum races}
In this appendix we finish the characterization of all Nash equilibria in a symmetric 2-player stingy quantum race.  We have 
seen in \cref{cor:interval} that a Nash equilibrium in a symmetric stingy quantum race is either coinciding, alternating, or 
alternating-coinciding, and in \cref{thm:symmetric} we have given the unique coinciding equilibrium.  
In this section, we characterize when alternating and alternating-coinciding equilibria exist. We begin by looking at alternating equilibria.

For the rest of this section, fix a stingy symmetric quantum race defined by probabilities 
$0  < p_1 < p_2 < \ldots < p_K \leq 1$, for some integer $K > 0$.

\subsection{Alternating equilibrium}
\label{sec:app_alt}
\subsubsection{At most one equilibrium}
We first consider the case of when an alternating equilibrium can exist.  To do this, it will be handy to introduce some notation.
For natural numbers $T < K$, we let $[T,K] = \{T, T+1, T+2, \ldots K\}$, and further let $D(T,K) = \{T, T+2, T+4, \ldots, K\}$ in the case $T,K$ have 
the same parity and $D(T,K) = \{T, T+2, T+4, \ldots, K-1\}$ in case $T,K$ have opposite parity.  We may assume, up to renaming of the 
players, that Alice plays at the smallest time.  Thus by \cref{cor:interval} we know that in an alternating equilibrium
Alice's support is $D(T,K) = \{T, T+2, T+4, \ldots, K-1\}$ and that Bob's support is $D(T+1,K) = \{T+1, T+3, \ldots, K\}$,
for some number $1 \leq T \leq K$, such that $K-T$ is odd. 

We define the values $\q_T$, for $T = 2, \ldots , K-1$, and the values $\er_T, \eR_T, \z_T$ 
and $\Z_T$, such that $1 \leq T < K$ and $K-T$ is odd,
as follows:
\begin{align*}
\q_T & =  \frac{1}{p_T} \cdot \left( \frac{1}{p_{T-1}} - \frac{1}{p_{T+1}} \right), \\
\er_T & =  \frac{1}{\bar p_T} \cdot \left( \frac{1}{p_K} - \sum_{i \in D(T+2,K)} \bar{p}_i \q_i  \right) , \\
\eR_T & = \frac{1}{p_{K-1}} - \sum_{i \in D(T+1,K-2)} \bar{p}_i \q_i , \\
\z_T & = \er_T + \sum_{i \in D(T+2,K)}  \q_i ,\\
\Z_T & = \eR_T + \sum_{j \in D(T+1,K-2)}  \q_j.
\end{align*}
\begin{definition}
\label{def:tstar_tilde}
Define $\eT^*$ as the smallest integer $1 \leq T < K$ such that $\er_T > 0$ and $K-T$ is odd. 
\end{definition}
This is a valid definition since $\er_{K-1} = 1/ (\bar{p}_{K-1} p_K) >0$. 
The following theorem states that there is at most one alternating Nash equilibrium in a symmetric stingy quantum race.

\begin{theorem}
\label{thm:alternating}
Every $2$-party symmetric stingy quantum race has exactly one alternating Nash equilibrium 
if  $\eR_{\eT^*} > 0$ and 
$p_K \left(\tfrac{1}{p_{K-1}} - p_K R_{\eT^*} \right) \le 1,$
and no alternating Nash equilibrium otherwise.
The supports of the equilibrium $(x,y)$ are $\sup(x) = D(\eT^*,K)$ and
$\sup(y) = D(\eT^*+1,K)$, 
and the distributions are defined as:

\begin{align*}
x_i & =
\begin{cases}
{\er_{\eT^*}}/{\z_{\eT^*}}& \text{if} ~~~ i= \eT^*,\\
{\q_i}/{\z_{\eT^*}}& \text{if} ~~~ i \in D(T+2,K),
\end{cases}
\\
y_j & =
\begin{cases}
{\eR_{\eT^* }}/{\Z_{\eT^*}}& \text{if} ~~~ j= K,\\
{\q_j}/{\Z_{\eT^*}}& \text{if} ~~~ j \in D(T+1,K-2).
\end{cases}
\end{align*}

\end{theorem}

\begin{proof}
As in the proof of \cref{lem:unicity}, instead of probability distributions we will normalize the strategies such that the payoff is 1.
We claim that the conditions of
Nash equilibrium for such a distribution can be expressed by the following system:
\begin{align}
e_i^{\T}Ay & =  1  & \text{for} ~~ i \in D(T,K) ~, \label{eq:ne11} \\
e_i^{\T}Ay & \leq  1  & \text{for} ~~ 1\leq i \leq K ~, \label{eq:ne12} \\
e_j^{\T}Ax & =  1  & \text{for} ~~ j \in D(T+1,K) ~, \label{eq:ne13} \\
e_j^{\T}Ax & \leq  1  & \text{for} ~~ 1\leq j \leq K ~, \label{eq:ne14} \\
0 <  & ~ ~x_i & \text{for} ~~ i \in D(T,K) ~, \label{eq:ne15} \\
0 <  & ~ ~y_j  & \text{for} ~~ j \in D(T+1,K) ~. \label{eq:ne16}  \\
x_i &= 0 & \text{ for } i \in [T-1] \cup D(T+1,K) \label{eq:alt_alice_pos}\\
y_j &= 0 & \text{ for } j \in [T-1] \cup D(T,K) \label{eq:alt_bob_pos}. \enspace
\end{align}
In matrix form \cref{eq:ne11} can be expressed as:
\begin{equation}
\begin{bmatrix}
1& 1 & 1 &  \cdots & 1 \\
\bar{p}_{T+1} & 1 & 1 & \cdots & 1 \\
\bar{p}_{T+1}  & \bar{p}_{T+3}  & 1 & \cdots  & 1 \\
\bar{p}_{T+1}  & \bar{p}_{T+3}  & \bar{p}_{T+5} & \cdots  & 1 \\
\vdots              &                        &                       &  \ddots  & \vdots \\
\bar{p}_{T+1}  & \bar{p}_{T+3} & \bar{p}_{T+5} & \cdots & 1 \\
\end{bmatrix}
\begin{bmatrix}
y_{T+1} \\ y_{T+3} \\ y_{T+5} \\ \vdots \\ y_K
\end{bmatrix}
=
\begin{bmatrix}
1/p_T \\ 1/p_{T+2} \\ 1/p_{T+4} \\ 1/p_{T+6} \\ \vdots \\  1/p_{K-1}
\end{bmatrix}.
\label{eq:alt_Alice}
\end{equation}
Simplifying the system gives
\[
\begin{bmatrix}
p_{T+1} & 0 & 0 &  \cdots & 0 \\
0            & p_{T+3} & 0 & \cdots & 0 \\
0            & 0            & p_{T+5} & \cdots  & 0 \\
0            & 0            & 0                & \cdots  & 0 \\
\vdots              &                        &                       &  \ddots  & \vdots \\
\bar{p}_{T+1}  & \bar{p}_{T+3} & \bar{p}_{T+5} & \cdots & 1 \\
\end{bmatrix}
\begin{bmatrix}
y_{T+1} \\ y_{T+3} \\ y_{T+5} \\ \vdots \\ y_K
\end{bmatrix}
=
\begin{bmatrix}
1/p_T - 1/p_{T+2} \\ 1/p_{T+2}- 1/p_{T+4} \\ 1/p_{T+4}-1/p_{T+6} \\ 1/p_{T+6}-1/p_{T+8} \\ \vdots \\  1/p_{K-1}
\end{bmatrix}.
\]
The unique solution is 
\begin{align*}
y_j &= \q_j, \qquad \mbox{for } j \in D(T+1,K-2), \\
y_K &= \eR_{T} \enspace.
\end{align*}
\cref{eq:ne13} in matrix form gives a very similar system:
\begin{equation}
\begin{bmatrix}
\bar{p}_T& 1 & 1 &  \cdots & 1 \\
\bar{p}_T & \bar{p}_{T+2} & 1 & \cdots & 1 \\
\bar{p}_T  & \bar{p}_{T+2}  & \bar{p}_{T+4} & \cdots  & 1 \\
\vdots              &                        &                       &  \ddots  & \vdots \\
\bar{p}_{T}  & \bar{p}_{T+2} & \bar{p}_{T+4} & \cdots & \bar{p}_{K-1} \\
\end{bmatrix}
\begin{bmatrix}
x_{T} \\ x_{T+2} \\ x_{T+4} \\ \vdots \\ x_{K-1}
\end{bmatrix}
=
\begin{bmatrix}
1/p_{T+1} \\ 1/p_{T+3} \\ 1/p_{T+5} \\ \vdots \\  1/p_{K}
\end{bmatrix}.
\label{eq:Bob_first}
\end{equation}
Similarly simplifying this system we have
\[
\begin{bmatrix}
0             & p_{T+2}          & 0 &  \cdots & 0 \\
0             & 0                     & p_{T+4} & \cdots & 0 \\
0             & 0  & 0 & \cdots  & 0 \\
\vdots              &                        &                       &  \ddots  & \vdots \\
\bar{p}_{T}  & \bar{p}_{T+2} & \bar{p}_{T+4} & \cdots & \bar{p}_{K-1} \\
\end{bmatrix}
\begin{bmatrix}
x_{T} \\ x_{T+2} \\ x_{T+4} \\ \vdots \\ x_{K-1}
\end{bmatrix}
=
\begin{bmatrix}
1/p_{T+1}-1/p_{T+3} \\ 1/p_{T+3}- 1/p_{T+5}\\ 1/p_{T+5}-1/p_{T+7} \\ \vdots \\  1/p_{K}
\end{bmatrix}.
\]
The unique solution is 
\begin{align*}
x_i &= \q_i, \qquad \mbox{for } i \in D(T+2,K), \\
x_T &= \er_T \enspace.
\end{align*}
Since $\q_t$ is positive for every $2 \leq t \leq K-1$, \cref{eq:ne15}--(\ref{eq:ne16})
are equivalent to
\begin{align}
\er_T  & >  0,   \label{eq:T} \\
\eR_{T}  & > 0.   \label{eq:K} 
\end{align}

We now check \cref{eq:ne12}.  As Alice's payoff for playing a time 
$i < T$ is maximized when $i = T-1$, we may assume that $i \in D(T-1,K)$.  Consider first
the case where Alice plays in $D(T-1,K-2)$; we will treat the case Alice plays time $K$ later.
When Alice plays $i \in D(T-1, K-2)$ then $i+1$ is in $\sup(x)$.  Moreover 
$D(T+1,K)_{\leq i} = D(T+1,K)_{\leq i+1}$.  
This means that the payoff for Alice playing $i+1 \in \sup(x)$ is strictly larger than that for playing $i$, specifically 
\begin{equation}
\label{eq:shift}
e_i^{\T}Ay =  \tfrac{p_i}{p_{i+1}} e_{i+1}^{\T}Ay <1.
\end{equation}
Thus Alice has no advantage in switching to the strategy $i$.  

We have now reduced \cref{eq:ne12} to the condition $e_K^{\T} A y \le 1$.  
We cannot use the preceeding argument to get rid of this condition as $K+1 \not \in \sup(x)$.  Let us 
now simplify this condition.  Substituting in $y$ gives the condition
\begin{equation}
\label{eq:A}
p_{K} \left(\sum_{j \in E(T) \setminus \{K\}} \bar{p}_j \q_j + \bar{p}_K \eR_{T} \right)\le 1 \enspace .
\end{equation}
Now observe that by \cref{eq:shift} with $i=K-2$,
\[
p_{K-2} \left(\sum_{j \in E(T) \setminus \{K\}} \bar{p}_j \q_j + \eR_{T} \right) = \frac{p_{K-2}}{p_{K-1}}\enspace .
\]
Using this equality shows that \cref{eq:A} is equivalent to
\begin{equation}
\label{eq:aliceK}
p_K \left(\tfrac{1}{p_{K-1}} - p_K R_{\eT^*} \right) \le 1 \enspace .
\end{equation}

Now consider \cref{eq:ne14}.  First say Bob plays a time $j \in D(T,K-1)$.  In this case $j+1 \in D(T+1,K)$.  
Similarly to Alice's case \cref{eq:shift},
$D(T,K-1)_{\leq j}= D(T,K-1)_{\leq j+1}$ and therefore 
$e_j^{\T}Ax =  \tfrac{p_j}{p_{j+1}} e_{j+1}^{\T}Ax <1.$
Thus Bob has no advantage in playing time $j$.
Finally, consider the case where Bob plays a time $j < T$.  In such a case, Bob's success probability is 
maximized by playing 
time $j = T-1$.  This gives the condition
\begin{equation}
\label{eq:B}
e_{T-1}^{\T}Ax = p_{T-1} \left( \er_T + \sum_{i \in D(T+2,K)} \q_i  \right)\le 1 \enspace .
\end{equation}

To summarize the current situation, \cref{eq:ne11}--(\ref{eq:alt_bob_pos}) are equivalent to 
\cref{eq:T}--(\ref{eq:K}),\ (\ref{eq:aliceK})--(\ref{eq:B}).
By the definition of $\eT^*$, \cref{eq:T} is satisfied if and only if $T \geq \eT^*$.

We claim now that \cref{eq:B} is satisfied if and only if $T \leq \eT^*$.
To prove this, first note that
\begin{equation}
\label{eq:rel}
{\bar{p}_{T-2}} \er_{T-2} = {\bar{p}_T} (\er_T - \q_T) \enspace.
\end{equation}
Using that, from \cref{eq:Bob_first}, applied to the case where the interval begins at $T-2$, we derive
\begin{equation}
p_{T-1} \left( \bar{p}_T \er_T + p_T \q_T + \sum_{i \in D(T+2,K)} \q_i  \right) =1.
\end{equation}
Therefore \cref{eq:B} holds if and only if $\er_T \leq \q_T$. From \cref{eq:rel} this is equivalent to
$\er_{T-2} \leq 0$, and the claim then follows from the definition of $\eT^*$.

To conclude, we have seen that there can not be more than one alternating Nash equilibrium, and there is
one for $T=T^*$ if \cref{eq:K} and~(\ref{eq:aliceK}) are satisfied, which is exactly the statement of the Theorem.
\end{proof}

\subsubsection{Starting value}
For this subsection we will impose a further natural restriction on the probabilities defining a quantum race that also holds for the 
Grover race.  We define a \emph{convex} quantum race.
\begin{definition}[Convex race]
Let $K$ be a positive integer. We say that the sequence  $(n_1, \ldots, n_K) \in \R^K$ is {\em convex}
if for all $1 < i < K$, we have $2n_i \leq n_{i-1} + n_{i+1}$.
A symmetric (stingy) quantum race specified by the  probabilities
$0  < p_1  < \ldots < p_K \leq 1$ is called {\em convex} if $( 1/p_1, \ldots ,  1/p_K )$ is a convex sequence.
\end{definition}

If $f : [1, K] \rightarrow \R$ is a convex real function on the interval $[1, K]$, and by definition 
$n_i = f(i)$, for $1 \leq i \leq K$, then clearly the sequence $(n_1, \ldots, n_K)$ is {convex}, which motivates the above definition.  
The (stingy) Grover race is convex as $\csc^2(2(t+1/2)\arcsin(\frac{1}{\sqrt{N}}))$ is a convex function.

For 2-player convex symmetric stingy quantum races we can show that the starting point of the support of the symmetric Nash 
equilibrium and the starting point of the support of the alternating Nash equilibrium (if it exists) are the same, modulo a 
correction in parity.
\begin{theorem}
Let $p_1, \ldots, p_K$ define a convex a 2-player symmetric stingy quantum race. 
\begin{enumerate}
\label{thm:equal}
\item 
If $K-T^*$ is odd then $\eT^* = T^*$,
\item If $K-T^*$ is even then $\eT^* \in \{T^*-1,T^*+1\}.$
\end{enumerate}
\end{theorem}
\begin{proof}
\cref{lem:les} below shows that $T^* \le \eT^* + 1$.  This direction does not make use of the 
convexity assumption.  \cref{lem:more} below shows $\eT^* \le T^* + 1$ and does make use of the convexity 
assumption.  From these two statements the Theorem follows.
\end{proof}

\begin{lemma}
\label{lem:les}
In every 2-player symmetric stingy quantum race we have $T^* \leq \eT^* + 1$.
\end{lemma}
\begin{proof}
We first claim that for every odd $i$ such that $3 \leq i \leq K-1$, we have
\begin{equation}
\label{eq:compar}
\bar{p}_i \q_i > \bar{p}_i q_i + \bar{p}_{i+1} q_{i+1} \enspace .
\end{equation}
For this first note that $p_i \q_i = p_i q_i + p_{i+1} q_{i+1}$.  
Thus the claim is equivalent to showing $\q_i  > q_i + q_{i+1}$,
which is is true since 
$\q_i - q_i  = \frac{p_{i+1}}{p_i} q_{i+1} > q_{i+1}.$


We now claim that for every $T$ such that $1\leq T \leq K-1$ and $K-T$ is odd, we have
\begin{equation}
\label{eq:compar2}
\er_T > 0 \Longrightarrow r_{T+1} > 0 \enspace,
\end{equation}
from which the statement of the Lemma clearly follows. This claim is clearly true for $T = K-1$, therefore
we suppose that $T \leq K-3$ and $K-T$ is odd. Summing up \cref{eq:compar} over $D(T+2,K)$ we 
get
\[
\sum_{i \in D(T+2,K)} \bar{p}_i \q_i > \sum_{i = T+2}^K 
\bar{p}_i q_i \enspace ,
\]
from which \cref{eq:compar2} immediately follows.
\end{proof}

\begin{claim}
\label{cla:hard}
In a 2-player convex symmetric stingy quantum race, for every $4 \leq i \leq K-2$, we have
\begin{equation}
\label{eq:pele}
\bar{p}_i q_i + \bar{p}_{i+1}q_{i+1} > \bar{p}_{i+1} \q_{i+1} \enspace .
\end{equation}

\end{claim}

\begin{proof}
Recall that $q_{i+1} - \q_{i+1} = - \tfrac{p_{i+2}}{p_{i+1}}q_{i+2}$.  Thus
\begin{align*}
\bar{p}_i q_i + \bar{p}_{i+1}q_{i+1} - \bar{p}_{i+1} \q_{i+1} &= 
\frac{\bar{p}_i}{p_i} p_i q_i - \frac{\bar{p}_{i+1}}{p_{i+1}} {p_{i+2}} q_{i+2} \\
&= \left(\frac{1}{p_i}-1 \right)\left(\frac{1}{p_{i-1}} - \frac{1}{p_i} \right) - \left(\frac{1}{p_{i+1}} -1 \right) 
\left(\frac{1}{p_{i+1}} - \frac{1}{p_{i+2}} \right) \\
& > \left(\frac{1}{p_{i+1}}-1 \right) \left(\frac{1}{p_{i-1}} - \frac{1}{p_i} -\frac{1}{p_{i+1}} + \frac{1}{p_{i+2}} \right) \\
& > \frac{1}{p_{i-1}} - \frac{1}{p_i} -\frac{1}{p_{i+1}} + \frac{1}{p_{i+2}}\\
& = \left(  \frac{1}{p_{i-1}} - \frac{2}{p_i} +\frac{1}{p_{i+1}}  \right)
+ \left( \frac{1}{p_i} -\frac{2}{p_{i+1}} + \frac{1}{p_{i+2}} \right) \\
& \ge 0 \enspace ,
\end{align*}
where the last inequality holds by convexity of $(1/p_1, \ldots, 1/p_K)$. 
\end{proof}

\begin{lemma}
\label{lem:more}
In a 2-player convex symmetric stingy quantum race we have $\eT^* \le T^*+1$.
\end{lemma}
\begin{proof}

The proof of the Lemma will follow from the following statement. 
For every $1\leq T \leq K-1$, we have:
\begin{align}
K-T \text{~is odd and } r_T > 0  ~ \Longrightarrow  ~
\er_{T} > 0,  \label{eq:egy} \\
 K-T \text{~is even and } r_T > 0  ~ \Longrightarrow  ~ 
 \er_{T+1} > 0. \label{eq:ketto}
\end{align}
\cref{eq:egy}  follows since $\er_T > r_T$ by summing over \cref{eq:pele}.
Similarly, we can see that $\tilde{r}_{T+1} > r_T$ by summing over \cref{eq:pele} from
$T+2$ until $K -1$:
$$
\sum_{i \in D(T+3,K)} \bar{p}_i \q_i < \sum_{i = T+2}^{K -1}
\bar{p}_i q_i ,
$$
from which \cref{eq:ketto} follows.
\end{proof}

\subsection{Alternating-coinciding Nash equilibria}
In this section, we turn to the third possible strategy structure for a Nash equilibrium, alternating-coinciding.  In this case, 
Alice and Bob play at alternating times from time $T$ until time 
$c-1$, at which point they play the same times in the interval $[c, K]$.  In other words, Alice's strategy is 
supported on a set $\mathcal{A}=\{T,T+2, \ldots, c-2, c, c+1, \ldots, K\}$ and Bob's strategy is supported on 
$\mathcal{B}= \{T+1, T+3, \ldots, c-1, c, c+1, \ldots, K\}$.  We call this a $(T,c,K)$-alternating-coinciding strategy, and 
refer to $c$ as the \emph{change point}.  
As we have already treated the coinciding and purely alternating cases, we will assume in this section that $T < c \le k$.  

To study alternating-coinciding Nash equilibria, we need to define the following quantities.
\begin{align*}
q_t &= \frac{1}{p_t}\left(\frac{1}{p_{t-1}}-\frac{1}{p_{t}} \right), \mbox{ for } 1 < t \le K \\
\tilde{q}_t &= \frac{1}{p_t}\left(\frac{1}{p_{t-1}}-\frac{1}{p_{t+1}} \right), \mbox{ for } 1 < t < K \\
r_{t,c} &= \frac{1}{\bar{p}_t} \left(\frac{1}{p_K} - \sum_{i \in D(t+2,c-2)} \bar{p}_i \tilde{q}_i - \sum_{i=c}^K \bar{p}_i q_i \right), \mbox{ for } 
1 \le t < c \le K \\
R_{t,c} &= \frac{1}{p_K} - \sum_{i \in D(t+1,c-3)} \bar{p}_i \tilde{q}_i - \sum_{i=c+1}^K \bar{p}_i q_i, \mbox{ for } 1 \le t < c \le K \\
w_{i,t,c} &= \bar{p}_{i}p_{c-1}\tilde{q}_{c-1} - p_i R_{t,c}, \mbox{ for } 1 \le t < i \le c \le K  \\
z_{t,c} &= r_{t,c} + \sum_{i \in D(t+2,c-2)} \tilde q_i + \sum_{i = c}^K q_i \\
Z_{t,c} &= \sum_{i \in D(t+1,c-3)} \tilde q_i -\frac{w_{c,T,c}}{p_c - p_{c-1}} + \frac{w_{c-1,T,c}}{p_c-p_{c-1}} + \sum_{i=c}^K q_i
\end{align*}

\begin{theorem}
\label{thm:alt_coinc}
Let $p_1 < p_2 < \cdots < p_K$ define a symmetric stingy quantum race.  Let $1 \le T < c \le K$ be such that $T,c$ are odd.  
There is a $(T,c,K)$-alternating-coinciding Nash equilibrium $(x,y)$ if and only if the following conditions are satisfied:
\begin{align}
\label{eq:bob_changes_thm}
\frac{1}{p_{T-1}} &\ge r_{t,c} + \sum_{i \in D(T+2,c-2)} \tilde{q}_i + \sum_{i=c}^K q_i \enspace,\\
\frac{1}{p_{c-1}} &\ge \sum_{i \in D(T+1,c-3)} \bar{p}_i \tilde{q}_i - \bar{p}_{c-1}\frac{w_{c,T,c}}{p_c - p_{c-1}} + 
\frac{w_{c-1,T,c}}{p_c-p_{c-1}} + \sum_{i=c+1}^K q_i \enspace, \\
\label{eq:first_pos}
\frac{1}{p_K} &> \sum_{i \in D(T+2,c-2)} \bar{p}_i \tilde{q}_i + \sum_{i=c}^K \bar{p}_i q_i\enspace, \\
\bar{p}_{c-1} \tilde{q}_{c-1} &> R_{T,c} \enspace.
\end{align}
If these constraints are satisfied then the Nash equilibrium is given by 
\begin{align*}
x_i = 
\begin{cases}
r_{T,c}/z_{T,c} & \mbox{ if } i =T \\
\tilde{q}_i/z_{T,c} & \mbox{ if } i \in D(T+2,c-2) \\
q_i/z_{T,c} & \mbox{ if } i \in [c,K] \\
0 & \mbox{ otherwise}
\end{cases}, \qquad
y_j = 
\begin{cases}
\tilde{q}_j/Z_{T,c} & \mbox{ if } j \in D(T+1,c-3) \\
-\frac{w_{c,T,c}}{Z_{T,c}(p_c - p_{c-1})} & \mbox{ if } j = c-1 \\
\frac{w_{c-1,T,c}}{Z_{T,c}(p_c-p_{c-1})} & \mbox{ if } j = c \\
q_j/Z_{T,c} & \mbox{ if } j \in [c,K] \\
0 & \mbox{ otherwise.}
\end{cases}
\end{align*}
\end{theorem}

\begin{proof}
Let $\mathcal{A} = D(T,c-2) \cup \{c, c+1, \ldots, K\}$ be Alice's support set and $\mathcal{B} = D(T+1,c-1) \cup \{c,c+1, \ldots, K\}$ be 
Bob's support set.  A Nash equilibrium $(x,y)$ with $x$ supported on $\mathcal{A}$ and $y$ supported on $\mathcal{B}$ must 
satisfy the equations
\begin{align*}
e_i^T A y = e_j^T A y \mbox{ for all } i,j \in \mathcal{A} \\
e_i^T A x = e_j^T A x \mbox{ for all } i,j \in \mathcal{B} 
\end{align*}
By multiplying $x$ and $y$ by appropriate constants, we can obtain $x', y'$ (now unnormalized) that satisfy 
\begin{align*}
e_i^T A y' = 1 \mbox{ for all } i \in \mathcal{A} \\
e_i^T A x' = 1 \mbox{ for all } i \in \mathcal{B} 
\end{align*}
These sytems of linear equations can be solved similarly to the systems in the symmetric and alternating cases, 
and always have a unique solution given by
\begin{align*}
x_i' = 
\begin{cases}
r_{T,c} & \mbox{ if } i =T \\
\tilde{q}_i & \mbox{ if } i \in D(T+2,c-2) \\
q_i& \mbox{ if } i \in [c,K] \\
0 & \mbox{ otherwise}
\end{cases}, \qquad
y_j' = 
\begin{cases}
\tilde{q}_j & \mbox{ if } j \in D(T+1,c-3) \\
-\frac{w_{c,T,c}}{p_c - p_{c-1}} & \mbox{ if } j = c-1 \\
\frac{w_{c-1,T,c}}{p_c-p_{c-1}} & \mbox{ if } j = c \\
q_j & \mbox{ if } j \in [c,K] \\
0 & \mbox{ otherwise.}
\end{cases}
\end{align*}
For (the normalized versions of) $x',y'$ to actually be a Nash equlibrium, $x',y'$ also need to satisfy the following additional 
inequalities.
\begin{align}
\label{eq:bob_changes}
e_j^{\T} A x' &\le 1 \mbox{ for } 1 \le j \le K \enspace,\\
\label{eq:alice_changes}
e_i^{\T} A y' &\le 1 \mbox{ for } 1 \le i \le K \enspace, \\
\label{eq:pos_alice}
0 < x_i' & \mbox{ for } i \in \mathcal{A} \enspace, \\
\label{eq:pos_bob}
0 < y_j' & \mbox{ for } j \in \mathcal{B} \enspace,
\end{align}
We now simplify these conditions given the values of $x',y'$, beginning with \cref{eq:bob_changes}.  With 
$x'$ defined as above, $e_j^{\T} A x' = p_j/p_{j+1}$ for 
$j \ge T, j \not \in \mathcal{B}$.  This is always at most one.  Thus for \cref{eq:bob_changes} what remains to be checked is 
that $e_{T-1}^{\T} A x' \le 1$.  

Similarly for \cref{eq:alice_changes} with $y'$ as defined above we see that $e_i^T A y' = p_i/p_{i+1}$ for $T-1 \le i < c-1, i 
\not \in \mathcal{A}$.   As playing a time $i < T-1$ will result in a strictly smaller payoff than playing time $T-1$, the only constraint 
that is not immediate is $e_{c-1}^T A y' \le 1$.  

Now for the positivity conditions, \cref{eq:pos_alice} and~(\ref{eq:pos_bob}).  It is clear that $\tilde{q}_i, q_i >0$ for $T < i \le K$ 
as the probabilities $p_j$ form an increasing sequence.  
The remaining positivity conditions that must be checked are $r_{T,c}, w_{c-1,T,c}, -w_{c,T,c} > 0$.  The condition 
$w_{c-1,T,c} >0$ is equivalent to $R_{T,c} < \bar{p}_{c-1} \tilde{q}_{c-1}$, and the condition $-w_{c,T,c} > 0$ is 
equivalent to 
\begin{align*}
R_{T,c} &> \bar{p}_c p_{c-1}\tilde{q}_{c-1}/p_c \\
&= \bar{p}_{c-1}\tilde{q}_{c-1}- \frac{p_c-p_{c-1}}{p_c}\tilde{q}_{c-1} \enspace . 
\end{align*}
This gives us the following simplified set of constraints.
\begin{align}
\label{eq:bob_changes_simple}
e_{T-1}^{\T} A x' \le 1 &\iff \frac{1}{p_{T-1}} \ge r_{t,c} + \sum_{i \in D(T+2,c-2)} \tilde{q}_i + \sum_{i=c}^K q_i \enspace,\\
\label{eq:alice_changes_simples}
e_{c-1}^T A y' \le 1 &\iff \frac{1}{p_{c-1}} \ge \sum_{i \in D(T+1,c-3)} \bar{p}_i \tilde{q}_i - \bar{p}_{c-1}\frac{w_{c,T,c}}{p_c - p_{c-1}} + 
\frac{w_{c-1,T,c}}{p_c-p_{c-1}} + \sum_{i=c+1}^K q_i \enspace, \\
\label{eq:alice_pos_simple}
 r_{T,c} > 0 &\iff \frac{1}{p_K} > \sum_{i \in D(T+2,c-2)} \bar{p}_i \tilde{q}_i + \sum_{i=c}^K \bar{p}_i q_i\enspace, \\
 \label{eq:bob_pos_simple}
w_{c-1,T,c}, -w_{c,T,c}  > 0,  &\iff \bar{p}_{c-1} \tilde{q}_{c-1}> R_{T,c} 
> \bar{p}_{c-1}\tilde{q}_{c-1}- \frac{p_c-p_{c-1}}{p_c}\tilde{q}_{c-1} \enspace.
\end{align}
We can further simplify the constraint in \cref{eq:bob_pos_simple}.  As $e_c^{\T}A y' =1$, we have 
\[
\frac{1}{p_{c}} = \sum_{i \in E(T+1,c-3)} \bar{p}_i \tilde{q}_i - \bar{p}_{c-1}\frac{w_{c,T,c}}{p_c - p_{c-1}} + 
\bar{p}_c\frac{w_{c-1,T,c}}{p_c-p_{c-1}} + \sum_{i=c+1}^K q_i \enspace .
\]
This means that \cref{eq:alice_changes_simples} is equivalent to 
\[
\frac{1}{p_{c-1}} - \frac{1}{p_c} \ge p_{c} \frac{w_{c-1,T,c}}{p_c-p_{c-1}}  \iff q_c \ge x_c'  \enspace .
\]
Rewriting this inequality in terms of $R_{T,c}$ gives
\begin{equation}
\label{eq:tighter_lb}
R_{T,c} \ge \bar{p}_{c-1} \tilde{q}_{c-1} - q_c\left(\frac{p_c - p_{c-1}}{p_{c-1}} \right)
\enspace.
\end{equation}
This constraint actually implies the lower bound of \cref{eq:bob_pos_simple}.  To see this, note that 
\begin{align*}
\frac{p_c-p_{c-1}}{p_{c-1}} q_c &= \frac{p_c-p_{c-1}}{p_{c-1}p_c} \left(\frac{1}{p_{c-1}}-\frac{1}{p_c} \right) \\
&< \frac{p_c-p_{c-1}}{p_{c-1}p_c} \left(\frac{1}{p_{c-2}}-\frac{1}{p_c} \right)\\
&= \frac{p_c-p_{c-1}}{p_c} \tilde{q}_{c-1} \enspace .
\end{align*}
Thus the lower bound in \cref{eq:bob_pos_simple} is a redundant constraint and can be removed.
This finishes the proof of the theorem.
\end{proof}

Next we examine which values of $T$ can be the starting point of an alternating-coinciding equilibrium. 
\begin{definition}
Let $T^*$ be as in \cref{def:tstar}.  For $c > T^*$ define 
$\hat T_c^* = \min\{T : r_{T,c} > 0\}$.
\end{definition}
The condition $c > T^*$ is taken so that $\hat T_c^*$ is always well defined.  The next claim shows that the 
starting point of an alternating-coinciding equilibrium with change point $c$ must be $\hat{T}_c^*$.
\begin{claim}
\label{clm:tstar}
Let $p_1 < p_2 < \cdots < p_K$ define a symmetric stingy quantum race.  For any $1 < c < K$, the starting 
support of an alternating-coinciding Nash equilibrium with change point $c$ (if it exists) must be 
$\hat{T}_c^* = \min\{T: r_{T,c} > 0\}$.  
\end{claim}

\begin{proof}
Note that $r_{T,c}$ is an increasing function of $T$.  This means that the condition $r_{T,c} > 0$ (\cref{eq:first_pos}) 
will only be satisfied for $\hat T_c^* \le T \le c$.  

We now see that to satisfy \cref{eq:bob_changes_thm} we must have $T \le \hat T_c^*$.
From the definition of $r_{t,c}$ we have
\[
\frac{\bar{p}_{t-2}}{\bar{p}_t} r_{t-2,c} = r_{t,c} - \tilde{q}_t \enspace .
\]
As $r_{\hat T_c^*-2,c} \le 0$ and $r_{\hat T_c^*,c} > 0$, this means that $r_{\hat T_c^*,c} \le \tilde{q}_{\hat T^*}$ and 
$r_{T,c} > \tilde{q}_{T}$ for $T > \hat T_c^*$.  

Considering when Bob plays strategy $T-1$ when the starting point of the interval is 
$T-2$, we see that 
\begin{align*}
\frac{1}{p_{T-1}} &= \bar{p}_{T-2} r_{T-2,c} + \sum_{i \in D(T,c-2)} \tilde{q}_i + \sum_{i=c}^K q_i \\
&= \bar{p}_T (r_{T,c} - \tilde{q}_T) +  \sum_{i \in D(T,c-2)} \tilde{q}_i + \sum_{i=c}^K q_i \\
&= \bar{p}_t r_{T,c} + p_T \tilde{q}_T + \sum_{i \in D(T+2,c-2)} \tilde{q}_i + \sum_{i=c}^K q_i
\end{align*}
From this equality we see that if $r_{T,c} > \tilde{q}_T$ then 
\[
\frac{1}{p_{T-1}} < r_{T,c} + \sum_{i \in O(T+2,c)} \tilde{q}_i +  \sum_{i=c}^K q_i
\]
 and so will violate \cref{eq:bob_changes_simple}.  
Thus the starting point of a alternating-coinciding Nash equilibrium with change point $c$ must satisfy $T \le \hat T_c^*$.  
Putting this together with $T \ge T_c^*$ means that the only possible starting point of a alternating-coinciding Nash equilibrium with 
change point $c$ is $\hat T_c^*$.
\end{proof}

\begin{claim}
\label{clm:sandwich}
Let $p_1 < p_2 < \cdots < p_K$ define a symmetric stingy quantum race.  Let $T^*, \tilde T^*$ be as in 
\cref{def:tstar},~\ref{def:tstar_tilde}, respectively.  For any $c > T^*$ we have 
$T^* -1 \le \hat T^*_c \le \tilde T^*+1$.  If $c$ and $K$ have the same parity then the tighter upper bound 
$\hat T^*_c \le \tilde T^*$ holds.
\end{claim}

\begin{proof}
Fix $c > T^*$.  For $T < c$ of the same parity as $c$, we show that if $r_{T,c} > 0$ then $r_{T+1} > 0$.  This implies 
$\hat T_c^* \ge T^*-1$.  

We have $r_{T,c} > 0$ if and only if
\begin{align*}
\frac{1}{p_K} &> \sum_{i \in D(T+2,c-2)} \bar{p}_i \tilde q + \sum_{i=c}^K \\
& \ge \sum_{i=T+2}^K \bar{p}_i q_i
\end{align*}
This inequality follows from the fact that $\bar{p}_i \tilde{q}_i \ge \bar{p}_i q_i + \bar{p}_{i+1} q_{i+1}$.
Thus we have $1/p_K > \sum_{i = T+2}^K \bar{p}_i q_i$, which implies $r_{T+1} > 0$.  

Now let $T,K$, and $c$ all have the same parity and suppose that $\tilde r_T >0$.  This happens if and only if 
\begin{align*}
\frac{1}{p_K} &> \sum_{i \in D(T+2, K)} \bar{p}_i \tilde{q}_i \\
& \ge \sum_{i \in D(T+2,c-2)} \bar{p}_i \tilde{q}_i + \sum_{i =c}^K \bar{p}_i q_i
\end{align*}
This means that $r_{T,c} > 0$.  This shows that if $c,K$ have the same parity then $\hat T^* \le \tilde T^*$.  If $c,K$ have 
opposite parity, then we can repeat the same argument to see that $\tilde r_T > 0 \implies r_{T+1,c} >0$.  Thus in 
the case we have $\hat T^*_c \le \tilde T^*+1$.
\end{proof}

We can analogously show that for every time $T$ there can be at most one change point $c$ such that there is an 
$(T,c,K)$-alternating-coinciding Nash equilibrium.  For this we first need the following claim.
\begin{claim}
\label{clm:tedious}
\[
\bar{p}_{t-1}\tilde{q}_{t-1} - q_t \left(\frac{p_{t}-p_{t-1}}{p_{t-1}}\right) = \bar{p}_{t-1}q_{t-1} + \bar{p}_t q_t \enspace .
\]
\end{claim}

\begin{proof}
First we expand the right hand side
\begin{align*}
\bar{p}_{t-1}q_{t-1} + \bar{p}_t q_t &= q_{t-1} - \frac{1}{p_{t-2}} + \frac{1}{p_{t-1}} + q_t - \frac{1}{p_{t-1}} + 
\frac{1}{p_t} \\
&= q_{t-1} +q_t + \frac{1}{p_t} - \frac{1}{p_{t-2}}   \enspace . 
\end{align*}
Next we expand the left hand side
\begin{align*}
\bar{p}_{t-1}\tilde{q}_{t-1} - q_t \frac{p_{t}-p_{t-1}}{p_{t-1}} = 
\tilde{q}_{t-1} - \frac{1}{p_{t-2}} + \frac{1}{p_t} - \frac{1}{p_{t-1}}\left(\frac{1}{p_{t-1}} - \frac{1}{p_{t}} \right) 
+q_{t} \enspace .
\end{align*}
Canceling like terms and rearranging, it suffices to show
\begin{align*}
\tilde{q}_{t-1} - q_{t-1} = \frac{1}{p_{t-1}}\left(\frac{1}{p_{t-1}} - \frac{1}{p_{t}}\right),
\end{align*}
which can be easily verified.
\end{proof}
\begin{lemma}
\label{clm:change_point}
For any time $T$ there is at most one $c$ such that there is a $(T,c,K)$-alternating-coinciding Nash equilibrium (up to 
switching the roles of Alice and Bob).
\end{lemma}
\begin{proof}

We show that for any $T$ there is at most one possible change point $c$ for an alternating-coinciding Nash equilibrium starting at $T$.
For this we look at \cref{eq:bob_pos_simple} and \cref{eq:tighter_lb}, which imply that the change point $c$ 
for an alternating-coinciding Nash equilibrium beginning at $T$ must satisfy
\begin{equation}
\label{eq:constraint}
\bar{p}_{c-1} \tilde{q}_{c-1}> R_{T,c} > \bar{p}_{c-1} \tilde{q}_{c-1} - q_c\left(\frac{p_c - p_{c-1}}{p_{c-1}} \right) \enspace .
\end{equation}
Using \cref{clm:tedious} we can rewrite \cref{eq:constraint} as 
\begin{equation}
\label{eq:constraint2}
\bar{p}_{c-1} \tilde{q}_{c-1}> R_{T,c} > \bar{p}_{c-1}q_{c-1} + \bar{p}_c q_c  \enspace .
\end{equation}
From the definitions, it can be verified that
\[
R_{T,t+2} - R_{T,t} = -\bar{p}_{t-1} \tilde{q}_{t-1} + \bar{p}_{t+1} q_{t+1} + \bar{p}_{t+2} q_{t+2} \enspace. 
\]
Now suppose that 
\begin{equation}
\bar{p}_{t-1} \tilde{q}_{t-1}> R_{T,t} > \bar{p}_{t-1}q_{t-1} + \bar{p}_t q_t  \enspace .
\end{equation}
Then 
\begin{align*}
R_{\hat T^*,t+2} &= R_{\hat T^*,t} -\bar{p}_{t-1} \tilde{q}_{t-1} + \bar{p}_{t+1} q_{t+1} + \bar{p}_{t+2} q_{t+2} \\
& < \bar{p}_{t+1} q_{t+1} + \bar{p}_{t+2} q_{t+2} \enspace .
\end{align*}
This means that if $R_{T,c}$ satisfies \cref{eq:constraint2}, then no $R_{T,t}$ can satisfy this constraint 
for $t \ge c+2$, as it will be smaller than corresponding lower bound in \cref{eq:constraint}.  This means that 
for each $T$, there is at most $c$ such that \cref{eq:constraint2} is satisfied.
\end{proof}

Finally, in the case of a convex quantum race we can more precisely pin down 
the exact form of an alternating-coinciding equilibrium.  If $T^*$ is the starting point of the coinciding Nash equilibrium, then 
the support of any alternating-coinciding Nash equilibrium must begin in the interval $[T^*-1,T^*+2]$.  
\begin{theorem}
Let $p_1 < p_2 < \cdots < p_K$ define a 2-player convex symmetric stingy quantum race. Let $T^*$ be the starting point of 
the symmetric Nash equilibrium.  Any alternating-coinciding equilibrium must 
begin at either $T^*-1, T^*,T^*+1$, or $T^*+2$.  In particular, there are at most four alternating-coinciding equilibria 
(up to switching the roles of Alice and Bob).
\end{theorem}

\begin{proof}
We first we show the statement about the starting point of the interval.
Let $\hat T^*_c = \min \{T: r_{T,c} > 0\}$.  \cref{clm:tstar} established that the starting point of a mixed Nash equilibrium with 
change point $c$ (if it exists) must be at $\hat T^*_c$.  By \cref{clm:sandwich}, we have that $T^*-1 \le \hat T^*_c \le \tilde T^*+1$ 
for any possible value of $c$.  This part of the theorem now follows as \cref{thm:equal} shows that for concave probabilities 
$\tilde T^* \in \{T^*-1,T^*, T^*+1\}$.  

The ``in particular" line of the theorem now follows by \cref{clm:change_point}.  For each $T \in [T^*-1, T^*+2]$ there can be at most 
one $c$ such that there is a $(T,c,K)$-alternating-coinciding Nash equilibrium.  Thus there can be at 
most four alternating-coinciding Nash equilibria (up to switching the roles of Alice and Bob).
\end{proof}

\end{document}